\documentclass{amsart}
\usepackage{graphicx} 
\usepackage{url,multirow, subcaption}
\usepackage{amsmath, amsthm, amsfonts, enumerate}
\usepackage{comment}
\newtheorem{example}{Example}

\newtheorem{proposition}{Proposition}

\begin{document}

\author{David McCune}
\address{David McCune, Division of Analytical Sciences, William Jewell College}
\email{mccuned@william.jewell.edu}

\author{Jennifer Wilson}
\address{Jennifer Wilson, Department of Natural Sciences and Mathematics, The New School}
\email{wilsonj@newschool.edu}
\title[Can ranked-choice voting elect the least popular candidates?]{Can ranked-choice voting elect the least popular candidate?}

\begin{abstract}
We analyze how frequently instant runoff voting (IRV) selects the weakest (or least popular) candidate in three-candidate elections. We consider four definitions of ``weakest candidate'': the Borda loser, the Bucklin loser, the candidate with the most last-place votes, and the candidate with minimum social utility. We determine the probability that IRV selects the weakest candidate under the impartial anonymous culture and impartial culture models of voter behavior, and use Monte Carlo simulations to estimate these probabilities under several spatial models. We also examine this question empirically using a large dataset of real elections.  Our results show that IRV can select the weakest candidates under each of these definitions, but such outcomes are generally rare. Across most models, the probability that IRV elects a given type of weakest candidate is at most 5\%. Larger probabilities arise only when the electorate is extremely polarized.\end{abstract}

\keywords{ranked-choice voting, instant runoff voting, weak candidate, spatial models}

\maketitle

\section{Introduction}

The purpose of this article is to explore the titular question, ``can ranked-choice voting elect the least popular (or weakest) candidate?'' The question was posed to the first author in informal discussion by someone interested in alternative voting rules. In that context, ranked-choice voting referred specifically to instant runoff voting (IRV), also known as the Hare method, the alternative vote, etc., which is the only ranked voting rule used for some U.S. political elections. A natural response to the question is \emph{no}, because IRV does not elect the Condorcet loser (the candidate who loses every pairwise contest), which is a standard definition of a weakest candidate in the social choice literature. Nor does IRV elect the candidate with the fewest first-place votes, another intuitive measure of unpopularity. However, other reasonable definitions exist, such as the candidate with the lowest Borda score. This article investigates how often IRV may elect such weak candidates under several different definitions of ``weakest.''

We focus on four definitions: the candidate with the lowest Borda score (the Borda loser), the candidate ranked last under Bucklin voting (the Bucklin loser), the candidate who receives the greatest number of last-place rankings, and the candidate with the lowest social utility. All terms are defined in the next section. These measures are chosen because IRV can elect such candidates, and because each is a reasonable measure of unpopularity. We do not claim that any of these definitions provides a definitive measure of candidate weakness. For example, it is possible for a strong candidate such as the Condorcet winner to have the lowest social utility. Our goal is therefore not to identify a single correct notion of weakness, but to assess the extent to which IRV may select candidates who fare poorly under a variety of sensible metrics.

We restrict our discussion to elections with three candidates and consider how often, and under what circumstances, the IRV winner corresponds to the weakest candidate under these different notions of weakness. We restrict to the three-candidate case because if an election contains four or more candidates, the likelihood that the weakest can win under IRV (or under any reasonable voting method) dramatically decreases. That is, the three-candidate setting is arguably the most interesting case. Our analysis combines analytical, numerical, and empirical tools because each framework captures different features of voter behavior. Together they allow us to distinguish outcomes that arise from random preference aggregation, ideological structure, and real-world voting behavior. We show that when voters have complete, single-peaked preferences, IRV cannot elect either the Bucklin loser or the candidate receiving the most last-place votes. More generally, we analyze the probability that the IRV winner is weak under several  models of voter behavior. Using analytical methods, we derive exact probabilities under the impartial culture and impartial anonymous culture models. We also employ Monte Carlo simulations for multiple one- and two-dimensional Euclidean spatial models to derive approximations for these models.

Because real-world IRV elections often contain ballots that do not rank all candidates, we  incorporate partial preferences into some of our simulations by allowing a substantial fraction of voters to cast \emph{bullet votes}, on which only a single candidate is ranked. (With three candidates, the only other kind of partial ballot is one that identifies first and second-place candidates, which we equate with ballots that rank the remaining candidate in third-place.) Allowing bullet votes allows us to better assess how ballot truncation affects the probability that IRV elects a weak candidate. 

Additionally, since the data from observed elections contains truncated ballots, we facilitate comparison between simulation and empirical analysis by completing the bullet votes in the real-world election data using a proportional completion procedure. By analyzing both the original ballot data and the completed data, we arrive at a more  consistent comparison between our models and what is observed in real elections.

Our results generally show that it is possible for IRV to select the weakest candidate under all of our notions of weakness. The frequency with which a weak candidate is selected is usually quite low, especially in real-world elections. There are notable exceptions: for example, under some of our spatial models, IRV elects the candidate with the lowest social utility in over 20\% of simulated elections. Percentages of this magnitude are rare, with most falling between 0\% and 5\%. These findings indicate that while IRV can, in principle, select candidates who perform poorly under several reasonable measures of popularity, such outcomes depend strongly on the structure of voter preferences and appear to be uncommon in real-world elections. Rather than serving as a blanket critique of IRV, these findings clarify the circumstances under which concerns about weak winners are most relevant.

To our knowledge, no previous study has examined this exact question in depth. We hypothesize there are two main reasons for this gap. First, the most common formal definition of a weakest candidate is the Condorcet loser, and (as stated above) IRV cannot elect this candidate. There is a large literature examining how often various voting methods elect the Condorcet loser (see \cite{GL98, L93, NU86}, for example), and of course this work does not address IRV. Second, much social choice research is framed around selecting the ``best'' candidate rather than focusing on avoiding the worst. Studies ask how often methods pick the Condorcet winner or the Borda winner, or how different rules encourage candidate moderation, for example. Modern scholarship surrounding these questions dates to the pioneering work of researchers such as Kenneth Arrow \cite{Arrow} and Duncan Black \cite{B48}, which has inspired a vast subsequent literature in social choice. 

Our work connects to three strands of previous literature. The first  consists of work involving \emph{a priori} models such as the impartial culture (IC) and impartial anonymous culture (IAC) models. This work uses analytical techniques to calculate exact probabilities under these models or uses simulations. See \cite{GL17, ST} for examples of calculations of Condorcet efficiency or paradoxical outcomes under IC or IAC. We provide results under these models in Section \ref{section:analytical_results}, restricting to the case where all voters provide a complete ranking of the candidates. 

The second strand of literature uses single-peaked or spatial models to investigate the behavior of voting rules. The notion of single-peakedness dates to the work of Black \cite{B48}, while modern study of spatial models begins with the work of Enelow and Hinich \cite{EH84, EH90}. Moulin \cite{M80} studies strategic voting when voter preferences are single-peaked, and Kamwa et al. \cite{KMT23} analyze how single-peakedness affects the frequency of various voting paradoxes. The use of spatial models has increased in conjunction with the rise of modern computing power. Plassmann and Tideman \cite{PT14} use two-dimensional spatial models to investigate the frequency with which different voting rules exhibit paradoxes. Holliday and Pacuit \cite{HP24} use several $k$-dimensional Euclidean models to estimate the expected social utility of election winners under different voting methods. Atkinson et al. \cite{AFG24} and Robinette \cite{R23} use one-dimensional models to analyze the tendency of Condorcet methods and IRV to elect moderate candidates. In Section \ref{section:simulation_results} we provide results using several different spatial models, for the case where all voters provide a complete candidate ranking and for the case where some voters do not.

The third strand consists of a rapidly expanding empirical literature enabled by the growing availability of ranked ballot data from jurisdictions such as Australia, Scotland, and the United States. This literature examines the Condorcet efficiency of IRV and other voting rules such as the Borda count \cite{FB25,GSM23,Song,S24}, the tendency of IRV to exhibit monotonicity paradoxes \cite{GSM23,GSZ20}, issues of ballot exhaustion \cite{BK15,TUK23}, etc. In Section \ref{section:empirical}, we provide results using a large dataset of real-world political elections from Australia, Scotland, and the U.S.

The paper is structured as follows. Section \ref{section:prelims} provides all definitions, accompanied by examples. In Section \ref{section:models_datasets} we describe our models and real-world datasets. Section \ref{section:analytical_results} provides our analytical results, giving exact probabilities under the impartial culture and impartial anonymous culture models. In Section \ref{section:simulation_results} we give the results obtained from simulations under our models using both complete and partial preferences. Section \ref{section:empirical} analyzes real-world datasets. In Section \ref{section:discussion} we discuss our results, primarily explaining why our models produce the results they do. To provide context, we also compare IRV to the methods of Condorcet and plurality. The article concludes in Section \ref{section:conclusion}.

\section{Preliminaries}\label{section:prelims}

IRV, under the more colloquial name ``ranked-choice voting,'' has been widely endorsed in the U.S. as an alternative to plurality voting. Advocates claim that the use of ranked ballots leads to more consensus-building during the campaign and to the election of more representative winners. Empirical studies indicate that IRV is relatively immune to spoiler candidates and will frequently select the Condorcet winner, when such a candidate exists \cite{schwab, S24}. IRV has an additional advantage in that it can never elect the plurality or Condorcet loser\footnote{Depending on one's exact definition of IRV, it may be possible for IRV to elect the Condorcet loser when there are ties in the elimination process. When a tie is broken at random, which is our approach, IRV can never elect the Condorcet loser.}. However, in some instances IRV may elect candidates who are weak under other notions of weakness, which we define below.

We begin with an explanation of the algorithm used to elect a winner under IRV. In an IRV election, voters cast preference ballots with a (possibly partial) linear ranking of the candidates. After the election, the ballots are aggregated into a \emph{preference profile}, which shows the number of each type of ballot cast. For example, Table \ref{table:small_profile_ex} shows a preference profile involving  three candidates $A$, $B$, and $C$. In the first column the number 6 means that 6 voters ranked $A$ first, $B$ second, and $C$ third, denoted $A \succ B \succ C$.  In real-world elections,  voters often provide only a partial ranking, casting a ballot such as $A\succ B$. For our purposes this ballot is equivalent to the complete ballot $A\succ B \succ C$, and thus the only partial ballots of interest are \emph{bullet votes}, on which only one candidate is ranked.

An IRV election proceeds in rounds. In each round, the number of first-place votes for each candidate is calculated. If a candidate has a majority of the (remaining) first-place votes, they are declared the winner. If no candidate receives a majority of votes, the candidate with the fewest first-place votes is eliminated and the names of the remaining candidates on the affected ballots are moved up. The process repeats until a candidate is declared the winner. In the case when partial ballots are allowed, any ballot in which all candidates have been eliminated is considered ``exhausted'' and plays no further role in the determination of the winner.  

We illustrate this algorithm for the preference profile in Table \ref{table:small_profile_ex} as follows.

\begin{table}
\centering
\begin{tabular}{l|cccc}
&6&8&4&1\\
\hline
1st&$A$&$B$&$C$&$C$\\
2nd&$B$&$C$&$A$&$B$\\
3rd&$C$&$A$&$B$&$A$\\
\end{tabular}
\caption{An example of a preference profile with 19 voters.}
\label{table:small_profile_ex}
\end{table}

\begin{example}\label{first_ex}
 In the first round,  we determine the number of first-place votes received by each candidate. In this case,  candidates $A$, $B$ and $C$ receive 6, 8, and 5 votes respectively. Since no candidate receives a majority, the candidate with the fewest first-place votes (the \emph{plurality loser}), $C$, is eliminated and their votes are transferred to the second-ranked candidates. Thus the  4 votes corresponding to the ranking $C \succ A \succ B$ are transferred to $A$ and the single vote corresponding to the ranking $C \succ B \succ A$  is transferred to $B$.  As a result, $A$                wins the election in the second round   with 10 votes to $B$'s 9 votes.
\end{example}

When there are only three candidates, it takes at most two rounds to determine the IRV winner.  If a candidate receives a majority of first-place votes in round 1, they are the winner. If no candidate receives a majority of first-place votes, then after a candidate is eliminated   in round 2, one of the remaining candidates receives a majority of the remaining votes (assuming there is not a tie) and becomes the winner. Thus, the IRV winner may be the \emph{plurality winner} (the candidate with the most first-place votes initially) or, as illustrated in  Example \ref{first_ex},  the runner-up in the first-place vote count.  


Regardless of the number of candidates, IRV cannot elect certain kinds of weak candidates. IRV will never elect the plurality loser, since they are eliminated in round 1. IRV will also never elect the Condorcet loser, the most commonly-used definition of a weak candidate. 

 Condorcet  winners and losers are defined based on the results of paired comparisons between candidates,  and can be determined whenever ranked ballots are used.  In Example \ref{first_ex}, for instance, if candidates $A$ and $C$ held a head-to-head contest, $C$ would win over $A$, 13 votes to 6, since each of the voters in the three rightmost columns prefer $C$ to $A$. A candidate that wins each of their head-to-head contests against the other candidates is called a Condorcet winner; a candidate that loses each of their head-to-head contests is   called a Condorcet loser. Condorcet winners and losers do not always exist; however, when they do, IRV  often (but not always) selects the Condorcet winner. However, in the absence of ties, IRV will never select the Condorcet loser because if such a candidate makes it to the final elimination round, they  will lose in a head-to-head contest to the final other remaining candidate. 

To investigate how often IRV elects other types of weak candidates,  we adapt  notions based on alternative voting methods.  The \emph{Borda count} is a positional scoring method in which candidates receive points based on their position on the ballot. If there are $m$ candidates, a candidate receives $m-1$ points for being ranked first, $m-2$ points for being ranked second, and so on. A candidate's Borda score is the sum of their points across all voters, and the candidate with the largest sum is the Borda winner. Since a lower Borda score indicates that a candidate ranks poorly overall in voters' ballots, we use the notion of a Borda loser (the candidate with the lowest Borda score) as one way to identify a weakest candidate. 

Our second definition of weak candidate is based on  Bucklin voting.  In the case of three candidates, Bucklin voting proceeds  as follows. If a candidate has a majority of first-place votes, they are declared a winner. Otherwise the number of second-place rankings for each candidate are added to their number of first-place rankings. If one candidate's sum surpasses a majority of votes, then that candidate is the winner; if multiple candidates' sums surpass a majority,  then the winner is the candidate with the largest sum. In the rare case that a winner is not declared at this stage (which can occur if the election contains many bullet votes), then the winner is determined in the same manner after adding each candidate's first-, second-, and third-place rankings. In the early 20th century, Buklin voting was used for political elections in several U.S. cities \cite{B2000}.

The \textbf{Bucklin loser} of an election is the candidate who has the lowest sum when the  winner is declared. For example, if a candidate has a majority of first-place votes then the Bucklin loser is the candidate with the fewest first-place votes. If no candidate has an initial majority but there is a candidate whose sum of first- and second-place rankings surpass a majority then the Bucklin loser is the candidate with the smallest sum of first- and second-place rankings.

\begin{example} (Example 1 revisited)
The IRV winner in Example \ref{first_ex} is $A$. However, $A$ is also the weakest candidate according to  both of the above measures. The Borda point totals for $A$, $B$, and $C$ are 16, 23, and 18, respectively. $A$ receives the fewest points and is thus the Borda loser. Under Bucklin voting, no candidate earns a majority of first-place votes, and thus we add each candidate's first-place rankings to their second, resulting in point totals of 10, 15, and 13 for $A$, $B$, and $C$, respectively. Thus, $A$ is also the Bucklin loser.
\end{example}

Our third definition of weak candidate is the \textbf{candidate with the largest number of last-place votes}. If all voters have complete preferences then this definition coincides with the Bucklin loser when there is no majority candidate. However, somewhat surprisingly, it is possible for a candidate to receive a majority of first place votes while still receiving the largest number of last-place votes.


The last definition of weak candidate is based on the idea of minimal social utility and is appropriate for spatial models in which voters and candidates are identified by positions in a finite-dimensional Euclidean space. If $\mathbf{c}$ and $\mathbf{v}$ are the position vectors for a candidate and a voter, respectively, we define the \textbf{utility provided to the voter by the candidate} as the negative of the Euclidean distance between them; i.e. the utility is $-||\mathbf{v}-\mathbf{c}||=-\sqrt{\sum (v_i-c_i)^2}$. 
There are many ways  to aggregate this function to measure  social utility, as nicely outlined in \cite{HP24}. To limit our scope, we choose, arguably, the simplest and  most natural: $-\displaystyle\sum_{\mathbf{v}}||\mathbf{v}-\mathbf{c}||$. We define the \textbf{candidate of minimal social utility} as the candidate that minimizes this sum.

Identifying the candidate with minimal social utility requires more information about voter preferences than that provided by a preference profile like in Table \ref{table:small_profile_ex}. Hence  we cannot apply this definition of weak candidate to all models of voter behavior, nor to real-world ranked ballot data.


When voters are allowed to cast bullet votes, some of these notions of weak candidate must be modified.  The definition of the Bucklin loser translates in the natural way to this setting, and the notion of minimal utility does not depend on ballot lengths. It is less clear how to adapt the notion of most last-place votes to preference profiles that include partial ballots. Moreover, this calculation loses its relevance when most voters cast bullet votes, as sometimes happens in our real elections. Hence we consider this definition only for elections with complete ballots. 

For the Borda loser, we must decide how to allocate Borda points for bullet votes. Following previous literature \cite{B12,K22}, we adapt the Borda scoring in three ways. Under the {\it optimistic model} (Borda OM), candidates not listed are each awarded 1 point (as if each were ranked second). Under the {\it pessimistic model} (Borda PM), candidates not listed are each awarded 0 points (as if each were ranked last). Under the {\it average model} (Borda AVG), any candidates not listed are awarded the average number of points available for being ranked second or third ($\frac{1+0}{2} =0.5$). 

In sum,  we consider four definitions of weak candidate:

\begin{itemize}
\item \textbf{The Borda loser.} The candidate with lowest Borda score. In the case of partial ballots, we use the Borda AVG, Borda OM, and Borda PM losers.
\item \textbf{The Bucklin loser.} The candidate who is ranked last by the Bucklin voting rule.
\item \textbf{The candidate with the most last place votes.} We calculate this candidate only for the case of complete ballots. 
\item \textbf{The candidate of minimal social utility.} We calculate this candidate only for spatial models, as it cannot be determined solely from a preference profile.
    
\end{itemize}

Why focus on these four particular definitions of weakest candidate? While other notions of least popular are sensible, the four used here arise naturally from well-known voting rules and concepts from social choice theory. The Borda count is a well-studied voting method in the social choice literature, and a candidate's Borda score is often used as a measure of the candidate's strength or ``representativeness.'' Thus, the candidate with minimal Borda score is perhaps the most intuitive notion of weakness aside from the Condorcet loser. Another prominent voting rule is approval voting, in which voters approve as many candidate as they wish and the candidate with the most approvals wins. In this setting, the candidate with the fewest approvals is arguably the weakest. Our elections use ranked rather than approval ballots and thus we cannot directly measure who would receive the fewest approvals, but the nature of Bucklin voting (in which first and second preferences are weighted equally) mimics the spirit of approval voting, at least in the three-candidate case. Thus, the Bucklin loser is a good proxy for an approval voting loser. Furthermore, as mentioned above, Bucklin voting has been used for political elections in the U.S. (but is not currently used to the best of our knowledge). 

Since IRV cannot elect the candidate with the fewest first-place votes we do not consider this measure, but the candidate with the most last-place votes is a natural mirror image measure. Finally, social utility is a common notion in economics and has been studied in the social choice context \cite{HP24, M84}. Additionally, social utility uses the ideological structure of voter preferences in ways that go beyond pure ranking information, which  allows for additional analysis of how weak candidates can be elected by IRV.  Thus, minimal social utility provides a useful complement to the three other measures.

Before beginning our analysis, we conclude this section with an observation. If voter preferences are {\emph single-peaked}, as they often are in  one-dimensional spatial models, then it is impossible for IRV  to elect either the Bucklin loser or the candidate with the most last-place votes.

The notion of single-peakedness dates to \cite{B48}. It assumes that candidates can be ordered along a horizontal axis in such a way that if each voter's preferences are graphed along a vertical axis, then each graph has a single local maximum. Informally, voter preferences are single-peaked if every voter agrees on which candidate is leftmost, which candidate is second leftmost, etc., and each voter's preferences respect these candidate placements. A consequence of single-peaked preferences is that not all candidate rankings are possible. With three candidates, if the candidates are ordered left-to-right,   $A \succ B \succ C$, then it is not possible for voters to have the rankings $A \succ C \succ B$ or $C \succ A \succ B.$
The one-dimensional spatial models we use are single-peaked, but the two-dimensional models are not.

\begin{proposition}
If voter preferences are single-peaked and voters cast complete ballots then the IRV winner cannot be the Bucklin loser, and cannot receive the most last-place votes.
\end{proposition}

\begin{proof}
Suppose the preference profile has the form shown in Table \ref{single-peaked}, and suppose that IRV selects the Bucklin loser. Then no candidate receives a majority of first-place votes. Moreover, since the sums of the first and second-place votes are  $a_1+b_1$, $a_1+b_1+b_2+c_1$, and $b_2+c_1$ for $A$, $B$, and $C$,  respectively, the Bucklin loser must be either $A$ or $C$. Suppose, without loss of generality, that $A$ is the IRV winner and the Bucklin loser. Then in the first round of IRV, either  $B$ or $C$ must be eliminated. If $C$ is eliminated then $B$ receives a total of $b_1+b_2+c_1$ votes, which must be a majority of votes cast (because $A$ does not have a majority initially). This implies that $B$ defeats $A$, which is impossible. If $B$ is eliminated, then in round 2, $A$ receives  $a_1+b_1$ votes  while $C$ receives $b_2+c_1$ votes. But $A$ is the Bucklin loser implying $a_1+b_1<b_2+c_1$, and hence $C$ defeats $A$, which is also impossible. Thus, $A$ cannot be both the IRV winner and the Bucklin loser.  
\begin{table}
\centering
\begin{tabular}{l|cccc}
&$a_1$&$b_1$&$b_2$&$c_1$\\
\hline
1st &$A$&$B$&$B$&$C$\\
2nd & $B$&$A$&$C$&$B$\\
3rd & $C$&$C$&$A$&$A$\\
 \end{tabular}

 \caption{An example of a single-peaked preference profile}
 \label{single-peaked}
  
  \end{table}

Next, suppose that IRV selects the candidate with the most last-place votes. Since $B$ receives no last-place votes, this means that IRV selects the candidate with a majority of last-place votes. But  this is impossible, since any candidate with a majority of last-place votes will lose all of its head-to-head contests, making it the Condorcet loser. \end{proof}

Elections with single-peaked preferences also have the characteristic that a Condorcet winner always exists-- 
the candidate closest to the median voter. In fact, assuming an odd number of voters,  it is always possible  to rank candidates based on the number of head-to-head contests they win. Thus, with three candidates, there is always a Condorcet winner who wins against both of the other candidates, a second candidate who wins  one head-to-head contest, and a Condorcet loser who does not win any head-to-head contests. Since IRV never elects the Condorcet loser, the IRV winner will either be the Condorcet winner or the candidate who wins over the Condorcet loser.  Instances where the latter occurs often correspond to situations in which the IRV winner is weak in one or more senses. We discuss this further in Section \ref{section:discussion}.

\section{Models and Datasets}\label{section:models_datasets}
To investigate the frequency with which IRV elects the least popular candidate, we use a  suite of models of voter behavior, together with a large collection of real-world ranked-choice elections. In this section we describe the models and data in detail.

\subsection{Description of Models}
 
  The first two models are  theoretical, \emph{a priori} models from the classical social choice literature.  They are constructed by placing a probability distribution on the set of possible ballots, using random selection in a natural way.  For both of these models, we assume that voters cast complete ballots. Thus, with three candidates, there are six possible rankings.  

The \textbf{impartial culture (IC) model} assumes each voter chooses randomly and uniformly from among all of the six possible rankings. Informally, each voter chooses a ranking by rolling a standard 6-sided die. Aggregating these rankings creates a profile like that in Table \ref{preference_profile}.  

The \textbf{impartial anonymous culture (IAC) model} assumes the preference profile is selected randomly and uniformly from among all possible preference profiles of a fixed size. Using the notation from Table \ref{preference_profile}, this means that if there  are $V$ voters, then the IAC model is generated by choosing at random a 6-tuple $(a_1, a_2, b_1, b_2, c_1, c_2)$, such that
$$a_1+a_2+b_1+b_2+c_1+c_2=V.$$

 \begin{table}
\centering
\begin{tabular}{l|cccccc}
&$a_1$&$a_2$ & $b_1$ &$b_2$ & $c_1$ &$c_2$\\
\hline
1st&$A$& $A$&$B$&$B$&$C$&$C$\\
2nd&$B$&$C$&$A$&$C$&$A$&$B$\\
3rd&$C$&$B$&$C$&$A$&$B$&$A$\\
\end{tabular}
\caption{A generic 3-candidate preference profile with complete ballots.} 
\label{preference_profile}
\end{table}

For both the IC and IAC models, there are analytical techniques that allow us to calculate the probabilities that specific events occur as $V \rightarrow \infty$. Thus, in  Section \ref{section:analytical_results}, where we discuss the IC and IAC models, we assume an arbitrarily large electorate.

We include the IC and IAC models in our analysis because they are prominent in the social choice literature, and they provide useful theoretical baselines for how IRV might behave in competitive elections generated from randomness. The IC model in particular generates elections that are extremely close in the sense that each possible ranking appears almost as frequently as any other. Consequently, it can be viewed as a ``worst-case'' environment in which unusual election outcomes are relatively likely to occur. The IAC model serves a similar role while allowing exact probability calculations using tools from geometric and combinatorial analysis. By comparing results from IC and IAC to those obtained from spatial models and real-world elections, we can better distinguish phenomena that arise from randomness from those that are driven by ideological structure or observed voting behavior. 

To complement these  \emph{a priori} models, we also consider spatial models that incorporate ideological structure among voters and candidates. Such models are useful because, unlike the IC and IAC models, or data obtained from real elections, spatial models allow for an analysis of social utility. We define three different kinds of one-dimensional spatial models, and build several different classes of two-dimensional models from the one-dimensional ones. Each model uses 4001 voters and voter preferences are constructed using Euclidean distance. That is, the 4001 voters and 3 candidates are placed in space, and each voter ranks the candidates based on their distance from the candidate. The candidate closest to them is ranked first, followed by the candidate next closest to them, and so on.  We use the number 4001 simply because it is large, mirroring the limiting probabilities we compute for the IC and IAC models.

Because analytical techniques do not exist for these spatial models, we use Monte Carlo simulations, discussed more thoroughly in Section \ref{section:simulation_results}, to investigate  how frequently specific events occur.

The spatial models differ in their dimension and in the probability distribution used to place the voters and the candidates.  We use the following probability distributions.
\begin{itemize}
\item \textbf{Unimodal, denoted UNI}. Voters and candidates are placed along the real line using the standard normal distribution.
\item \textbf{Bimodal with parameter} $\mathbf{\sigma}$, \textbf{denoted BIM}$(\sigma)$. For each $\sigma \in \{0.25, 0.26, \dots, 0.70\}$, voters and candidates are placed along the real line using a distribution which is a mixture of a normal distribution centered at $-1$ and a normal distribution centered at 1, each with standard deviation $\sigma$. The two normal distributions are weighted equally.
\item \textbf{Weighted bimodal with parameter} $\mathbf{\sigma}$, \textbf{denoted WBI}$(\sigma)$. This model is identical to BIM$(\sigma)$, except that 60\% of voters are drawn from the left distribution and 40\% from the right. This model allows for an analysis of polarized electorates where one wing of the electorate is larger than the other.
\end{itemize}

We include models such as BIM$(\sigma)$ and WBI$(\sigma)$ to capture varying degrees of political polarization, a prominent feature of contemporary American and international politics. This allows us to investigate the effect of polarization on IRV's propensity to elect the least popular candidate, and theoretical models like IC and IAC are not appropriate for this purpose.

Figure \ref{figure:1D_models} shows sample voter distributions under each model in the one-dimensional case. As $\sigma$ decreases, the electorate becomes more polarized, as illustrated for both BIM and WBI in Figure~\ref{figure:1D_models}. As $\sigma$ increases, the distributions approach a unimodal normal shape, though WBI remains skewed. Since we already include the model UNI, we cut off the $\sigma$ values at $0.70$.

We construct two-dimensional models by assigning coordinates in $\mathbb{R}^2$ to each voter and candidate using two one-dimensional models. Given one-dimensional models $X$ and $Y$, we denote the corresponding two-dimensional model by $X \times Y$. To generate such a model, we draw $x$-coordinates according to the distribution specified by $X$ and $y$-coordinates according to the distribution specified by $Y$, then combine these coordinates to form points in $\mathbb{R}^2$. For example, the model BIM$(\sigma)\times$WBI$(\sigma)$ generates $x$-coordinates from the bimodal distribution BIM$(\sigma)$ and $y$-coordinates from the weighted bimodal distribution WBI$(\sigma)$. See  Figure~\ref{figure:2D_models} for examples of voter distributions produced by some of the 2D models. The resulting electorates exhibit varying degrees and forms of polarization depending on the underlying one-dimensional distributions and may contain multiple dense ideological clusters. For example, models such as BIM$(\sigma)\times$BIM$(\sigma)$ produce electorates concentrated in two polarized clusters, while other combinations may generate more complex structures. We study such models because our primary interest is understanding how ideological polarization affects the tendency of IRV to elect weak candidates. As demonstrated in Section  \ref{section:simulation_results}, two-dimensional polarized electorates sometimes produce much different results than one-dimensional polarized electorates. Voter preferences are then determined by Euclidean distance to candidates, as in the one-dimensional case.

As mentioned previously, $k$-dimensional Euclidean models have been widely used in the social choice literature. Our specific models are inspired by those used in \cite{HP24}. Most previous work with spatial models assumes voters provide complete preferences; we run simulations where this is the case but also include complementary simulations where some voters cast bullet votes. This allows us to separate effects driven by preference structure from those driven by ballot truncation, which is common in real-world IRV elections. To generate simulated elections with partial preferences, we give each voter a probability of 0.35 of casting a bullet vote. We choose this number to match the ``typical'' rate of bullet votes in American ranked-choice elections, a dataset described below. The exact methodology for the simulations is described in Section \ref{section:simulation_results}.

\begin{figure}
\begin{tabular}{cc}
\includegraphics[width=76mm]{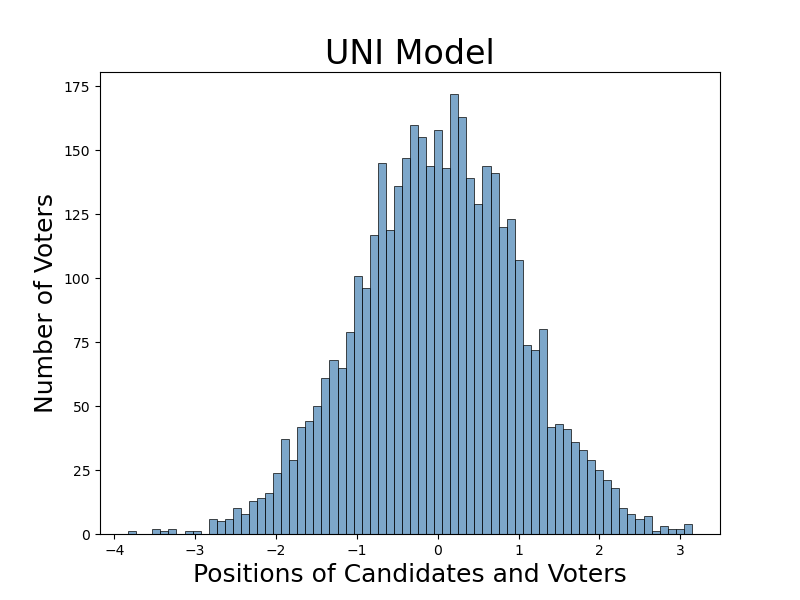} &\includegraphics[width=76mm]{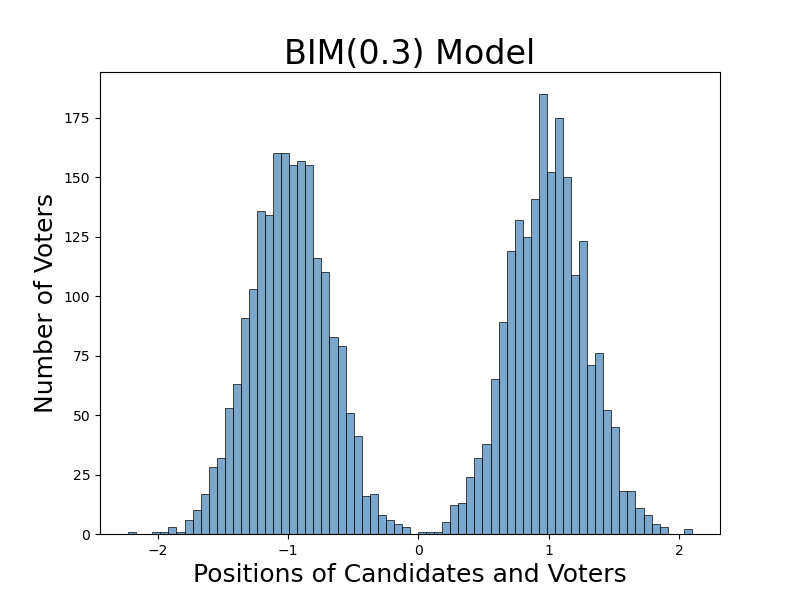} \\

\includegraphics[width=76mm]{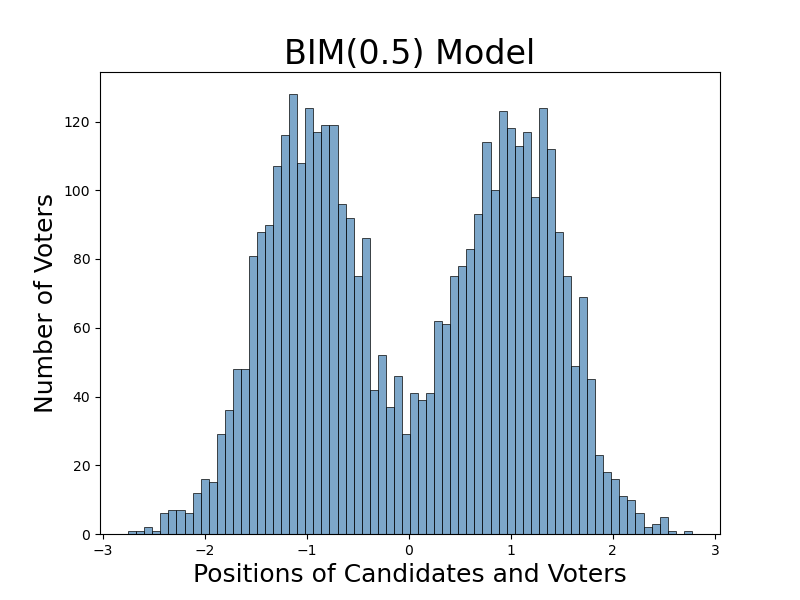}&\includegraphics[width=76mm]{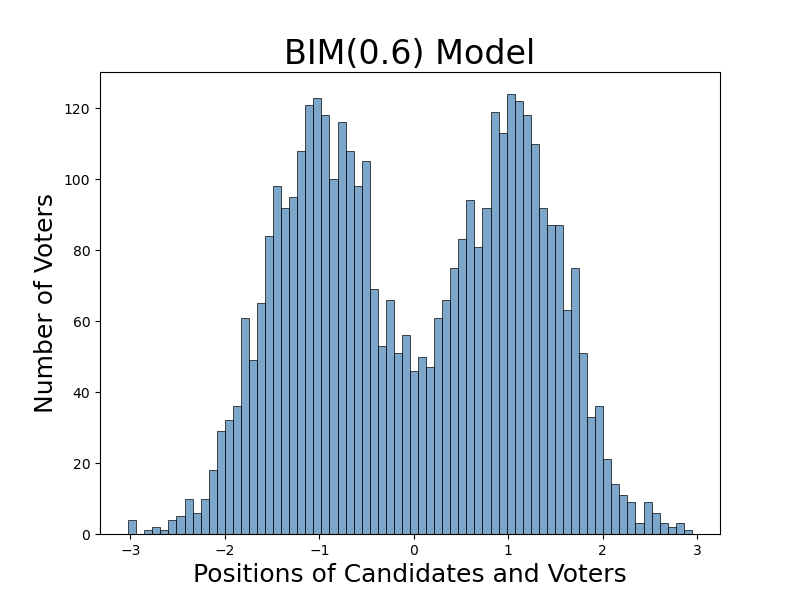}\\

\includegraphics[width=76mm]{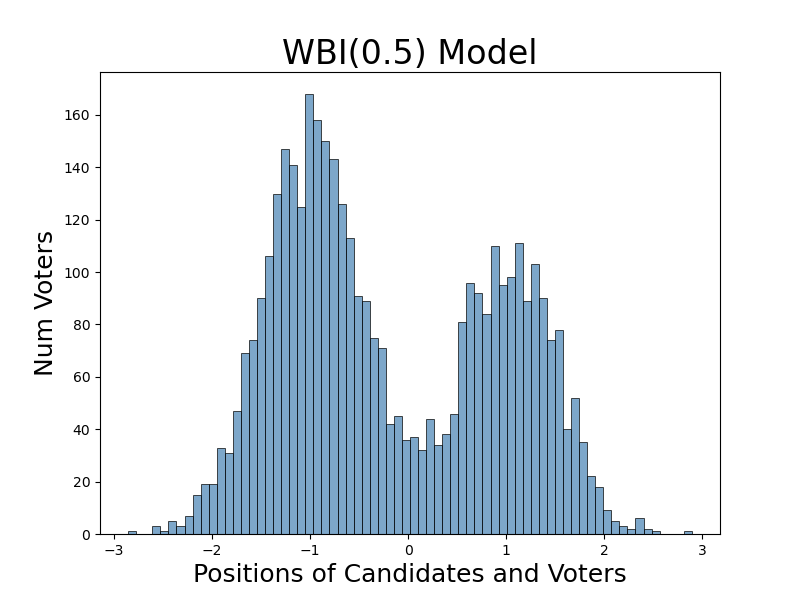}&\includegraphics[width=76mm]{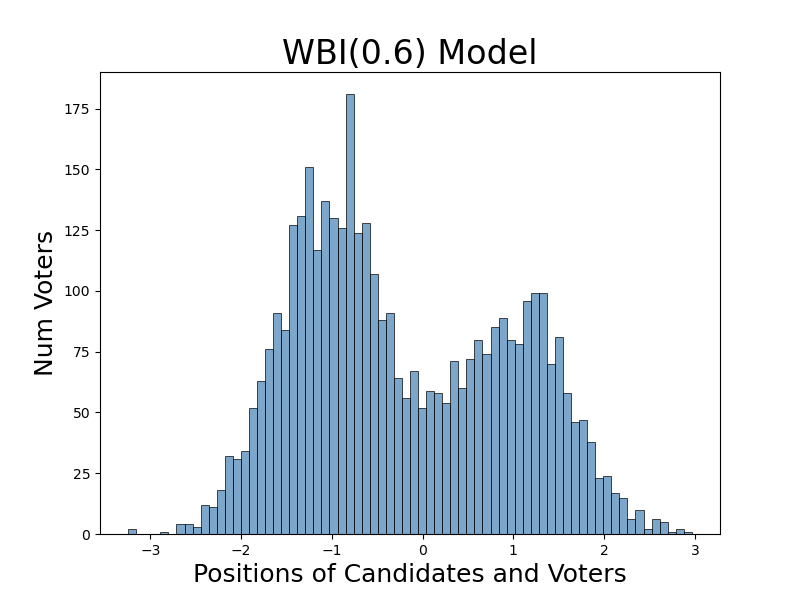}\\

\end{tabular}
\caption{Examples of one-dimensional spatial voter distributions under the UNI, BIM($\sigma$), and WBI($\sigma$) models, illustrating increasing polarization as $\sigma$ decreases.}
\label{figure:1D_models}
\end{figure}

\begin{figure}
\begin{tabular}{cc}

\includegraphics[width=76mm]{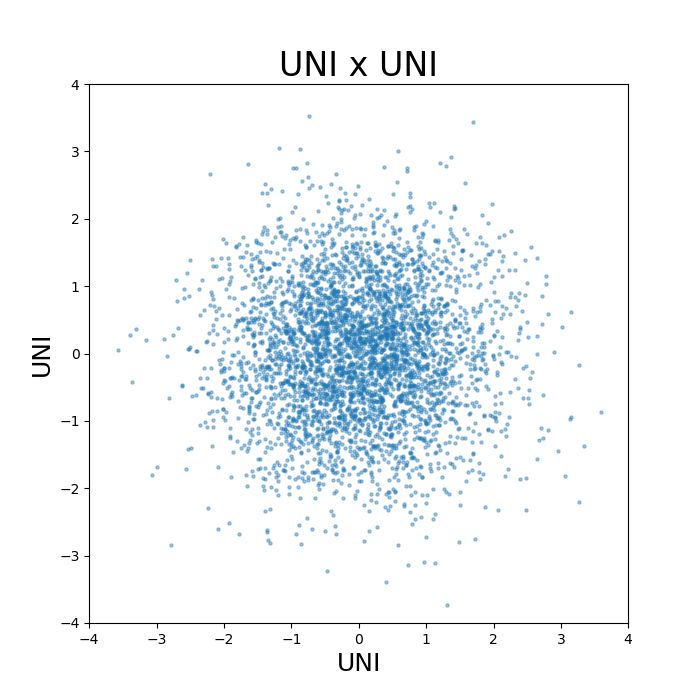}&\includegraphics[width=76mm]{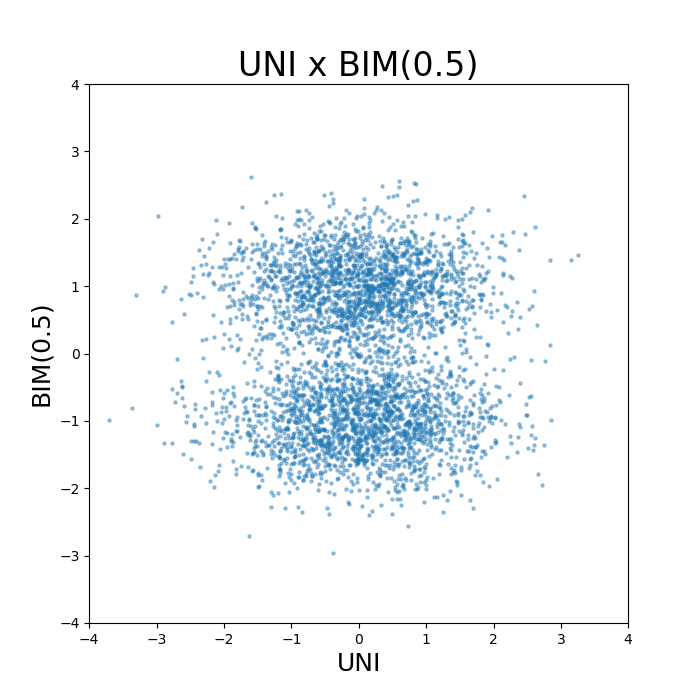}\\

\includegraphics[width=76mm]{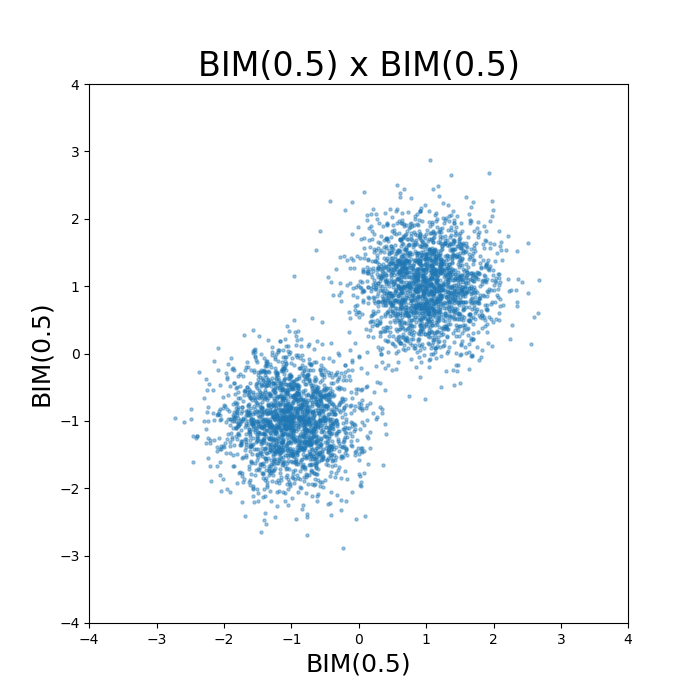}&\includegraphics[width=76mm]{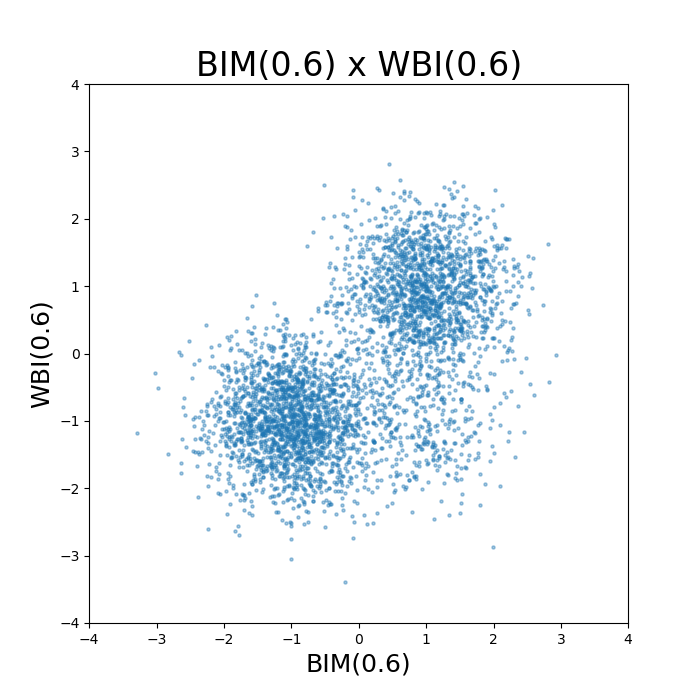}\\

\end{tabular}
\caption{Examples of four 2D spatial models.}
\label{figure:2D_models}
\end{figure}

In addition to simulated elections generated from our models, we analyze four real-world datasets drawn from \cite{schwab}. We also supplement these datasets with 27 ranked elections from the American Psychological Association (APA). Some of these APA elections were analyzed in \cite{PPR}, and others were provided directly to the first author for \cite{MM24}. The datasets we use are described below.

\begin{itemize}
 
\item \texttt{Australia}. This dataset consists of 325 single-winner political elections from Australia, primarily local elections in New South Wales. It was first analyzed in \cite{S24}.
\item \texttt{Scotland}. This dataset consists of 1128 local government elections from Scotland. Most of these elections are multiwinner, where the method of single transferable vote (a multiwinner proportional version of IRV) is used to elect 3 or 4 winners. The ballot data for 1100 of these elections is available at \url{https://github.com/mggg/scot-elex}. The data for the remaining 28 elections are available by request (these elections are off-cycle by-elections). This dataset has been analyzed in (\cite{BBMP24,BBDGW21,GSJM24,MGS24,Mollison23}).
\item \texttt{Scot\_condensed}. The dataset \texttt{Scotland} consists mostly of multiwinner elections, and ranked ballot data from multiwinner contests looks fundamentally different than data from single-winner elections. For example, in multiwinner elections parties often run multiple candidates in an attempt to win multiple seats. To make the Scottish dataset look more ``single-winner-like,'' we use the dataset \texttt{Scot\_condensed}, which was created for \cite{schwab}. For each election in \texttt{Scotland}, exactly one candidate per party was retained on each ballot, with all other candidates from that party removed. The candidate to retain was chosen as the party’s highest-ranked candidate under a Borda scoring rule, reflecting the candidate with the broadest within-party support.
\item \texttt{USA}. This dataset consists of the ballot data for 401 American ranked-choice elections compiled from \cite{L22a,L22b,MM24,O22}. Most are municipal elections in cities such as Minneapolis, Oakland, and New York City, though the dataset also includes federal elections in Alaska and Maine, as well as 27 elections from the APA. We restrict to elections which contain at least three official (not write-in) candidates. Elections from this dataset have been analyzed in \cite{FB25, GSM23,HOAG24,MM24,MW}.
    
\end{itemize}

Most elections in these datasets contain more than three candidates and therefore do not directly fit our analysis. The number of elections with exactly three candidates (excluding write-in candidates) is: 6 in \texttt{Australia}; 3 in \texttt{Scotland}; 43 in \texttt{Scot\_condensed}; and 133 in \texttt{USA}, for a total of only 185 elections.
To incorporate elections with more than three candidates, we apply the IRV algorithm to each election until only three candidates remain, and then analyze the resulting preference profile. This mirrors the elimination process inherent to IRV and preserves the relative ordering information among the remaining candidates. Such a procedure has drawbacks: for example, the resulting profiles tend to contain more bullet votes than are observed in actual three-candidate elections. Nonetheless, given our restriction to three candidates, this approach provides a reasonable way to leverage the full set of available data. 

To provide a direct comparison to our simulation results that assume complete voter preferences, we also analyze real-world preference profiles where bullet votes are extended to complete preferences using a proportional methodology. The details are described in Section \ref{section:empirical}.

\section{Analytical Results}\label{section:analytical_results}

In this section, we determine the probability that IRV selects a weak candidate under the IAC and IC models. We consider only complete preferences and, since neither model is spatial, confine the analysis to the frequency of selecting the Borda loser, Bucklin loser, or candidate with the most last-place votes.

To determine these probabilities, we identify the algebraic conditions under which one of these events can occur. Using the notation introduced in Table \ref{preference_profile}, we assume without loss of generality that the plurality scores rank the candidates $A \succ B \succ C$.

Consider first the possibility that the IRV winner is also the Borda loser. We claim that the IRV winner cannot have a majority of first-place votes. To see this, suppose $A$ has a majority of first-place votes and the smallest Borda score. Then
\begin{align*}
a_1+a_2 +\frac{1}{2}(b_1+c_1) &< b_1+b_2 +\frac{1}{2}(a_1+c_2)\mbox{ and }\\
a_1+a_2 +\frac{1}{2}(b_1+c_1) &< c_1+c_2 +\frac{1}{2}(a_2+b_2).
\end{align*}
Summing these inequalities, we obtain
$$2(a_1+a_2)+b_1+c_1< b_1+b_2+c_1+c_2+\frac{1}{2}(a_1+a_2+b_2+c_2),$$
which is equivalent to
$$a_1+a_2<b_2+c_2,$$
which proves the claim.

Thus, there are two ways in which the IRV winner and Borda loser coincide, corresponding to the two sets of inequalities (\ref{a_wins}) and (\ref{b_wins}) below: (\ref{a_wins}) $A$ does not have a majority, $A$ is the IRV winner, and $A$ is the Borda loser; and (\ref{b_wins}) $B$ is the IRV winner and the Borda loser. Therefore, the IRV winner and Borda loser are equal if and only if $a_1+a_2>b_1+b_2>c_1+c_2$ and additionally, either:

\begin{align}
&\left\{ \begin{array}{rl}
a_1+a_2+c_1  &>b_1+b_2+c_2 \mbox{ and} \\
a_1+a_2 + \frac{1}{2}(b_1+c_1)
&< \min\{ b_1+b_2 + \frac{1}{2}(a_1+c_2),
          c_1+c_2 + \frac{1}{2}(a_2+b_2)\}
\end{array}\right.
\label{a_wins}\\[1ex]
&\mbox{or}\nonumber \\[1ex]
&\left\{ \begin{array}{rl}
a_1+a_2+c_1 &<b_1+b_2+c_2 \mbox{ and} \\
b_1+b_2 + \frac{1}{2}(a_1+c_2)
&< \min\{ a_1+a_2 + \frac{1}{2}(b_1+c_1),
          c_1+c_2 + \frac{1}{2}(a_2+b_2)\}
\end{array}\right.
\label{b_wins}
\end{align}

Unlike with the Borda loser, it is possible for the IRV winner to have a majority of first-place votes and also receive the largest number of last-place votes. Thus, the IRV winner can receive the most last-place votes in three ways, corresponding to the three sets of inequalities, (\ref{a_majority}), (\ref{a_round2}) and (\ref{b_round2})  below (where we still assume $a_1+a_2>b_1+b_2>c_1+c_2$): (i) $A$ receives a majority of first-place votes; (ii) $A$ does not receive a majority of first-place votes but is the IRV winner; or (ii) $B$ is the IRV winner.

\begin{align}
&\left\{ \begin{array}{rl} a_1+a_2&>b_1+b_2+c_1+c_2 \mbox{ and}\\
b_2+c_2 &> \max\{ a_2+c_1, a_1+b_1\}, \end{array}\right. \label{a_majority}\\
&\mbox{or}\nonumber \\
& \left\{ \begin{array}{rl}
b_1+b_2+c_1+c_2&>a_1+a_2, \\
a_1+a_2+c_1&>b_1+b_2+c_2, \mbox{ and}\\
b_2+c_2 &> \max\{ a_2+c_1, a_1+b_1\},\end{array} \right.\label{a_round2} \\
&\mbox{or} \nonumber\\
& \left\{ \begin{array}{rl}
a_1+a_2+c_1&<b_1+b_2+c_2, \mbox{ and} \\
a_2+c_1 &> \max\{ a_1+b_1, b_2+c_2\}. \end{array}\right.\label{b_round2}
\end{align}

Since a candidate receives the most last-place votes exactly when they receive the least first and second-place votes (with complete ballots), inequalities (\ref{a_round2}) and (\ref{b_round2}) also describe the conditions under which the IRV winner and the Bucklin loser coincide.

Given these sets of inequalities, limiting probabilities (when $V \rightarrow \infty$) involving  IAC can be determined using  the theory of Ehrhart polynomials, which has been implemented in the software package Normaliz (Bruns et al., 2022). To explain briefly, each set  of inequalities defines a sub-simplex of the 6-simplex $a_1+a_2+b_1+b_2+c_1+c_2=1$. Ehrhart polynomials provide a tool for calculating their limiting volumes, which, since the IAC model places a uniform distribution on the simplex, is equivalent to calculating the corresponding probabilities. See the classic paper \cite{LLS} for a more in-depth discussion of the connection between the theory of Ehrhart polynomials and voting theory.

We use Normaliz to obtain the following IAC results.

\begin{proposition}
Under the IAC model, in three-candidate elections where voters provide complete preferences, the limiting probability (as $V \rightarrow \infty$) that the IRV winner coincides with the 

\begin{enumerate}
    \item Borda loser is $4301/241920 = 1.78\%$.
    \item Bucklin loser is $115/2304 =4.99\%$.
    \item candidate with the most last-place votes is $49/768=6.38\%$.
    \end{enumerate} 
\end{proposition}

To determine analogous probabilities for the IC model,  we adapt the geometric methods originally introduced in \cite{ST} and later refined in \cite{KM} and \cite{KMT23}, which also provide limiting probabilities as $V \rightarrow \infty$. Elections generated under the IC model are typically much different than those generated under the IAC model; hence the techniques used for calculating IC probabilities are also much different. The IC method involves integrals that often must be   evaluated numerically. Thus, the results  in  Proposition \ref{IC_prob} are (very close) approximations.

Note that the IC model typically produces  preference profiles that are extremely competitive (the numbers of ballots containing each ranking are very close). Thus, the probability that any candidate achieves a majority of first-place votes is vanishingly small and approaches zero as $V\rightarrow \infty$. So, we can assume that the  inequalities in (\ref{a_majority})  do not occur, which implies the probability that the IRV winner coincides with the Bucklin loser is equal to the probability that the IRV winner coincides  with the candidate with the most  last-place votes.

\begin{proposition} \label{IC_prob} 
Under the IC model, in three-candidate elections where voters provide complete preferences, the limiting probability (as $V\rightarrow \infty$) that the IRV winner coincides with the 
\begin{enumerate}
\item   Borda loser is approximately 2.811\%.
\item Bucklin loser or the  candidate with the most last-place votes is approximately 7.178\%.
\end{enumerate}
\end{proposition}

The proof is lengthy and involves a nice application of the Gauss-Bonnet Theorem. We include   part  of the proof of (1) here; the remainder of the proof is sketched in Appendix 2. 

 \begin{proof} 

Assume the order of first-place votes is $A \succ B \succ C.$ Let $A \succ_C B$ (respectively $B \succ_C A$)  denote that $A$ wins over $B$ (respectively $B$ wins over $A$) if $C$ is eliminated, and let $BSc(A) > BSc(B)$ mean that $A$ has a higher Borda score than $B$.  
We  divide each calculation into two cases, based on whether $A$ or $B$ is the IRV winner, and sum the results. 
Each case takes several steps.
 
\vskip 2mm

\noindent \textbf{Proof of (i) Step 1} Finding $P(A\succ B \succ C, B\succ_C  A)$, or the probability that $A\succ B \succ C$ and $B\succ_C  A$:

These conditions can be expressed as:
\begin{align*}
a_1+a_2 &> b_1+b_2,\\
b_1+b_2 &> c_1+c_2,  \mbox{   and} \\
b_1+b_2+c_2&>a_1+a_2+c_1.
  \end{align*}

Following the arguments of \cite{ST}, the probability of this occurring for
a large number of voters  is equal to the area of the spherical simplex $S$ defined by these
three inequalities on the surface of the unit sphere in $\mathbb{R}^3$, divided by the area of this sphere.
Let
\begin{align*}
 \mathbf v_1 &= (1, 1, -1, -1, 0, 0),  \\
 \mathbf v_2 & = (0,0, 1, 1, -1, -1),  \mbox{   and}\\
  \mathbf v_3 &= (-1,-1, 1, 1, -1, 1) 
  \end{align*}
   be the normal vectors of the three hyperplanes bounding $S$. By the Gauss-Bonnet Theorem,  
   $$Vol_2 (S) = \alpha_{12}+  \alpha_{13}+ \alpha_{23}-\pi$$ where $\alpha_{ij}$ is the angle between vectors $\mathbf v_i$ and $\mathbf v_j$ and $Vol_n$ is the $n-$dimensional volume. 
Using the law of cosines, $\alpha_{12} = \cos^{-1} \left(-\frac{\mathbf v_1 \cdot  \mathbf v_2}{||\mathbf v_1|| \cdot || \mathbf v_2 ||}\right) =  \cos^{-1} (\frac{2}{4}) = \frac{\pi}{3}.$ Similarly,  $\alpha_{13}=\cos^{-1}( \frac{4}{2\sqrt{6}})=\cos^{-1}( \frac{2}{\sqrt{6}})$ and   $\alpha_{23}= \cos^{-1}(-\frac{2}{2 \sqrt{6}})=  \cos^{-1}(-\frac{1}{\sqrt{6}})$. Thus, the probability is equal to approximately
  \begin{equation*}
\frac{1}{4 \pi}Vol_2 (S) = \frac{1}{4\pi} \left[ \frac{\pi}{3}+\cos^{-1}( \frac{2}{\sqrt{6}}) +  \cos^{-1}(-\frac{1}{\sqrt{6}}) - \pi \right] =0.0407767114.  \end{equation*}

\vskip 2mm
\noindent \textbf{Step 2} Finding $P(A\succ B \succ C, B\succ_C  A, \mbox{ and } BSc(A)>BSc(B))$, or the probability that additionally, $A$'s Borda score is higher that $B$'s Borda score:

This requires the additional condition $a_1+a_2+\frac{1}{2}(b_1+c_1)  > b_1+b_2+\frac{1}{2}(a_1+c_2).$
 Following \cite{KM}, we introduce the  modified inequality with parameter $t$,
   \begin{equation*} a_1+a_2+\frac{t}{2}(b_1+c_1)  > b_1+b_2+\frac{t}{2}(a_1+c_2)  \quad \text{with normal vector} \quad  \mathbf v_4 = (1-\frac{t}{2}, 1, \frac{t}{2}-1, -1, \frac{t}{2}, - \frac{t}{2}).
 \end{equation*}  
If $t=0$, the inequality reduces to a previous one; if $t=1$, we recover the required condition.  Similarly to Step 1, the probability that these conditions are met for a large number of voters  is equal to the  volume of the spherical simplex $S$  defined by these four inequalities on the surface of the unit sphere in $\mathbb{R}^4$, divided by the volume of the surface of this sphere, $2\pi^2$. 
To find  $Vol_3(S)$, we use the Sch{\"a}fli formula 
\begin{equation*}
d Vol_n = \frac{1}{n-1} \sum_{1 \le j < k \le n} Vol_{n-2}(S_j \cap S_k) d\alpha_{jk},
\end{equation*}
where $n=3$ and $S_i$ are the hyperplanes bounding $S$. Let $S_i$ be the hyperplane with normal vector $\mathbf v_i$.

Then as before,
\begin{equation*}
\alpha_{14}= \cos^{-1}\left(-\frac{(1, 1, -1, -1, 0, 0) \cdot (1-\frac{t}{2}, 1, \frac{t}{2}-1, -1,  \frac{t}{2}, - \frac{t}{2}) }{ 2\sqrt{4-2t+t^2}}\right)=\cos^{-1}\left(\frac{t-4}{2\sqrt{4-2t+t^2}}\right).
\end{equation*}
Similarly $\alpha_{24}  =\cos^{-1}\left(\frac{4-t}{4\sqrt{4-2t+t^2}}\right)$ and $\alpha_{34} = \cos^{-1}\left(\frac{2\sqrt{2}}{ \sqrt{3}\sqrt{4-2t+t^2}}\right).$
This leads to 
\begin{equation*}
d\alpha_{14} =\frac{ -\sqrt{3} }{ (t^2-2t+4) },  \quad d\alpha_{24}=\frac{ \sqrt{3}t }{ (t^2-2t+4) \sqrt{5t^2-8t+16} } \quad \text{and} \quad d \alpha_{34}= \frac{ 2\sqrt{2}(t-1)}{(t^2-2t+4)\sqrt{3t^2-6t+4}}. \end{equation*}

To calculate $Vol_1(S_j \cap S_k)$, we must find where the hyperplanes intersect.  First, we  identify a basis for the subspace orthogonal to that spanned by $\mathbf v_1 \ldots, \mathbf v_4$. We use $\mathbf v_5= (-1, 1, -1, 1, 0, 0) $ and $\mathbf v_6 = (1,1,1,1,1,1)$.  Let $P_{ijk}$ be the vertex lying at  the intersection of $S_i$, $S_j$ and $S_k$.  To find $P_{124}$, we solve the following conditions:
\begin{align*}
n_1+n_2 - n_3-n_4=0 \\
n_3+n_4 - n_5-n_6=0 \\
-n_1-n_2+n_3 + n_4-n_5+n_6>0 \\
(1-\frac{t}{2})n_1+n_2+(\frac{t}{2}-1)n_3-n_4+ \frac{t}{2}n_5- \frac{t}{2}n_6=0 \\
-n_1+n_2-n_3+n_4=0\\
n_1+n_2+n_3+ n_4+n_5+n_6=0.\\
\end{align*}
This yields $P_{124}=(-1, 1, 1, -1, -1, 1)$. In the same way, we find, $P_{134}=(1, 1, 1, 1, -2, -2) $ and $P_{234} = (-5t+12, 13t-12, 7t-12, -11t+12, -8t, 4t)$.
Thus,

 \begin{align*}
 Vol_{1}(S_1 \cap S_4)= \angle(P_{124}, P_{134})&=\cos^{-1}\left( \frac{P_{124} \cdot P_{134}}{\sqrt{6}\cdot 2\sqrt{3} }\right)=\cos^{-1}\left(0 \right)=\frac{\pi}{2} \\
Vol_{1}(S_2 \cap S_4 )=\angle(P_{124}, P_{234}) &=\cos^{-1}\left( \frac{P_{124} \cdot P_{234}}{\sqrt{6}\cdot 2\sqrt{3}\sqrt{37t^2-72t+48}   }\right)\\
& =\cos^{-1}\left(\frac{48(t-1)}{6\sqrt{2}\sqrt{37t^2-72t+48}} \right)=\cos^{-1}\left(\frac{4\sqrt{2}(t-1)}{\sqrt{37t^2-72t+48}} \right) \\
 Vol_{1}(S_3 \cap S_4)= \angle(P_{134}, P_{234}) &=\cos^{-1}\left( \frac{P_{134} \cdot P_{234}}{2\sqrt{3}\cdot 2 \sqrt{3} \sqrt{37t^2-72t+48}   }\right)\\
 &=\cos^{-1}\left(\frac{12t}{12\sqrt{37t^2-72t+48}} \right)= \cos^{-1}\left(\frac{t}{\sqrt{37t^2-72t+48}}\right).
 \end{align*}
Substituting these values into the Sch{\"a}fli formula with $n=3$,  we get
$$d Vol_3 = \frac{1}{2} \sum_{1 \le j < k \le 4} Vol_{1}(S_j \cap S_k) d\alpha_{jk},$$ 
so
\begin{align*}
Vol_3(S_{t=1}) & = Vol_{3}(S_{t=0})+ \int_{0}^1 d Vol_3  \\
& =   Vol_{3}(S_{t=0}) +\frac{1}{2} \int_{0}^1  Vol_{1}(S_1 \cap S_4) d\alpha_{14}+Vol_{1}(S_2 \cap S_4) d\alpha_{24}+Vol_{1}(S_3 \cap S_4) d\alpha_{34} dt \\
& =  Vol_{3}(S_{t=0}) +\frac{1}{2}\left[I_1+I_2+I_3 \right]
 \end{align*}
where, up to 10 decimal places, 
$I_1 = \int_{0}^1  Vol_{1}(S_1 \cap S_4) d\alpha_{14}= \int^1_0\frac{\pi}{2} \cdot \frac{ -\sqrt{3} }{ (t^2-2t+4) } dt = -0.82246703,$
$I_2 =0.1508811848,$ and $I_3 =-0.4079803509.$

By Step 1,  $Vol_{3}(S_{t=0}) = Vol_2(S_{t=0})\frac{2\pi^2}{4\pi} =0.0407767114 (2\pi^2)$. Hence, the probability that $B$ is the IRV winner and has a lower Borda score than $A$ is approximately
\begin{equation*}
\frac{1}{2\pi^2} \left[0.0407767114(2\pi^2) + \frac{1}{2}( -0.82246703)+0.1508811848-0.4079803509) \right] = 0.01342486871.
\end{equation*}

The remainder of the proof is similar and continues in Appendix 2. \end{proof}


\section{Simulation Methodology and Results}\label{section:simulation_results}
In this section we describe the methodology of our simulations and present the results.

For each model, we ran 100,000 simulations with $V=4001$ voters and three candidates. Voters and candidates were placed either on the real line or in $\mathbb{R}^2$ according to the specified probability distribution for each model. Voter preferences were determined  by increasing Euclidean distance to the candidates, with ties among these distances broken uniformly at random. We also ran simulations with $V=1001$ and $V=3001$ voters; our results were essentially unchanged. We use an odd number of voters to minimize the probability of ties in vote tallies (with the infrequent ties broken at random), and we use a large enough number of voters to mimic sizable electorates in real elections (and to mirror the limiting probabilities presented in the previous section).

To incorporate partial preferences, we separately ran simulations using the same methodology but where each voter has a probability of 0.35 of casting a bullet vote. The number 0.35 is based on partial preference data from the dataset \texttt{USA}. In American  three-candidate ranked-choice elections, the median \emph{bullet vote rate} (the proportion of voters who rank only one candidate in a given election) is approximately 0.35, assuming there is no candidate who earns more than 70\% of the first-place vote.\footnote{If there is a candidate who earns a strong majority of the first-place votes then there are usually many more bullet votes cast, presumably because most voters are aware their favorite is the only viable candidate.} Of course, there are other sensible ways to incorporate partial preferences, such as using ideological distance cutoffs. We choose random truncation because of its simplicity.

The simulation results for both complete and partial ballots are presented in Table \ref{table:sim_results_1D} and in Figures \ref{figure:1D_results}, \ref{figure:2D_complete_ballots}, \ref{figure:2D_partial_ballots}, and \ref{figure:cond_prob_Borda_is_minimal_utility} (see Appendix 1 for the last three figures). Table \ref{table:sim_results_1D} gives a sample of the simulation results for 8 models, where each table entry is the percentage of simulated elections in which IRV elects the corresponding kind of weak candidate. For example, when some voters cast bullet votes the estimated probability IRV elects the Borda AVG loser is 1.0\% under the UNI model. The figures provide complete results for models involving BIM($\sigma$) or WBI($\sigma$). Figure \ref{figure:1D_results} shows the percentage of simulated elections in which BIM($\sigma$) or WBI($\sigma$) chooses a weak candidate for a given $\sigma$. Figures \ref{figure:2D_complete_ballots} and \ref{figure:2D_partial_ballots} provide similar information for the 2D models.

Overall, the probability that IRV elects a weak candidate using any of the definitions considered here, or under any model, is fairly low. (Note that the probability of IRV selecting the Bucklin loser or the candidate with the most last-placed votes is 0 under all one-dimensional spatial models, as asserted in Proposition 1.)  The highest probabilities for complete ballots in Table \ref{table:sim_results_1D}  are 8.53\% and  12.08\% for the candidate with the most last-place votes and the Borda loser respectively. If bullet votes are allowed, the highest probability is 21.2\% that the IRV winner is the candidate with the minimum utility. All of these occur under the BIM(0.5) $\times$ BIM(0.5) model, suggesting that IRV is most likely to choose the weakest candidate when an electorate is hyper-polarized.

With some exceptions, Figures \ref{figure:1D_results}, \ref{figure:2D_complete_ballots}, and \ref{figure:2D_partial_ballots} indicate that the probability that IRV selects a weak candidate are larger for more polarized electorates (smaller $\sigma$). 
The figures also highlight that the Borda and Bucklin losers and the candidate of minimal social utility are often different. In particular, when partial ballots are introduced, IRV rarely elects a Borda loser, but it frequently selects the candidate of minimal utility under BIM($\sigma$) for small $\sigma$ (with the exception of the UNI $\times$ BIM($\sigma$) model). Arguably, the worst-case scenario occurs when the Borda loser is also the candidate of minimal utility, and IRV elects this candidate. To quantify this, we computed the percentage of elections in which IRV chooses the Borda loser conditional on the outcome that the Borda loser is the minimal-utility candidate; the results are shown in Figure \ref{figure:cond_prob_Borda_is_minimal_utility}. For complete ballots, this conditional probability is substantial only in highly polarized BIM electorates, and otherwise remains very low. Overall, IRV rarely produces this extreme outcome in our simulations.

\begin{table}
\centering
\begin{tabular}{c|cccc}
\multicolumn{5}{c}{\textbf{Complete Ballots}}\\
\hline
\hline
Model & Min Utility & Borda loser & Bucklin loser & Most last place\\
\hline

UNI &0.19\% &1.74\% &0\% &0\% \\
BIM(0.5) & 0.69\%&8.86\% &0\% &0\% \\

WBI(0.5)&3.68\% &1.68\% &0\% &0\% \\

UNI$\times$UNI &0.11\% &0.33\% & 0.89\%& 1.27\%\\
UNI$\times$BIM(0.5)&0.42\%&1.34\% & 3.18\%&4.58\%\\

UNI$\times$WBI(0.5)&1.41\% & 0.57\% & 1.30\%&2.34\%\\

BIM(0.5)$\times$BIM(0.5)&2.70\% & 12.08\%&2.15\%&8.53\%\\

WBI(0.5)$\times$WBI(0.5)&3.59\%&0.4\%&0.25\%&1.91\%\\
\hline
IC&$-$&$2.81$\%&7.18\%&7.18\%\\
IAC&$-$&1.78\%&4.99\%&6.38\%\\
\hline

\end{tabular}

\vspace{.2 in}

\begin{tabular}{c|ccccc}
\multicolumn{6}{c}{\textbf{Partial Ballots}}\\
\hline
\hline
Model & Min. Utility & Borda AVG & Borda OM & Borda PM & 
 Bucklin \\
 \hline

 UNI &1.4\% &1.0\% &0.9\% &1.0\% &0.0\%  \\
 BIM(0.5) & 9.8\%& 2.3\%&2.4\% &2.2\% & 0.0\% \\
 WBI(0.5) & 3.7\%&1.7\% &1.4\% &1.9\% &0.0\%  \\
 UNI$\times$UNI & 0.5\%& 0.2\%& 0.2\%&0.2\% &0.5\%  \\
 UNI$\times$BIM(0.5) &2.8\% &0.6\% &0.6\% &0.7\% &2.5\%  \\
 UNI$\times$WBI(0.5)&1.9\% &0.5\% &0.4\% &0.5\% &0.9\%  \\
 BIM(0.5)$\times$BIM(0.5)&21.2\% &0.3\% &0.4\% &0.3\% &3.7\%  \\
 WBI(0.5)$\times$WBI(0.5)&3.9\% &0.6\% &0.4\% &1.0\% &0.2\%  \\
 \hline
\end{tabular}
\caption{A sample of results under our spatial models. Each entry is the percentage of simulated elections in which IRV elects the a specific type of weak candidate, under a given model. For comparative purposes, we include the IC and IAC probabilities from Section \ref{section:analytical_results} for the case of complete ballots.}
\label{table:sim_results_1D}
\end{table}

\begin{figure}[htbp]
    \centering

    \begin{minipage}{0.485\textwidth}
        \centering
        \includegraphics[width=\linewidth]{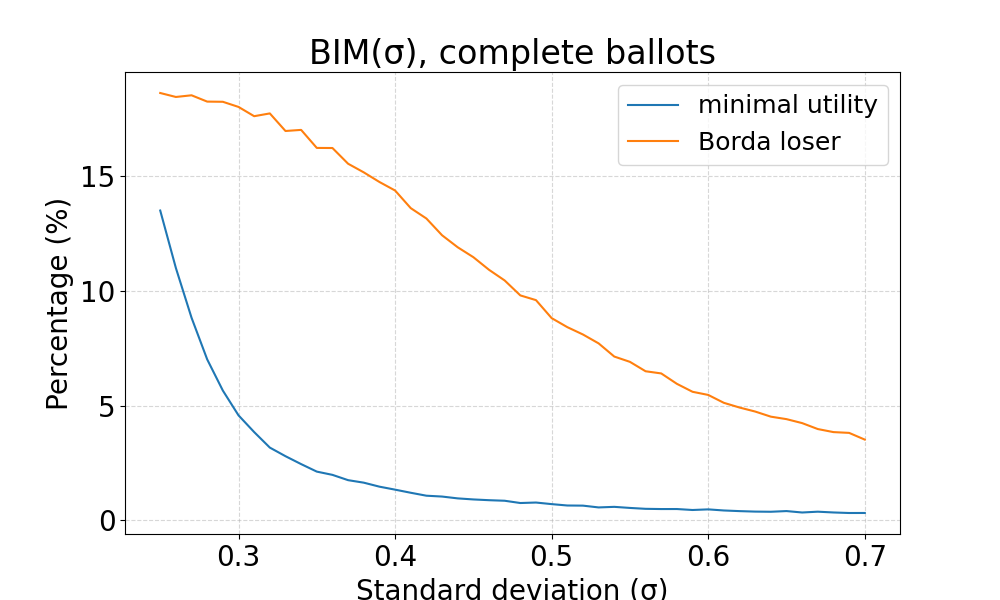}
    \end{minipage}
    \begin{minipage}{0.485\textwidth}
        \centering
        \includegraphics[width=\linewidth]{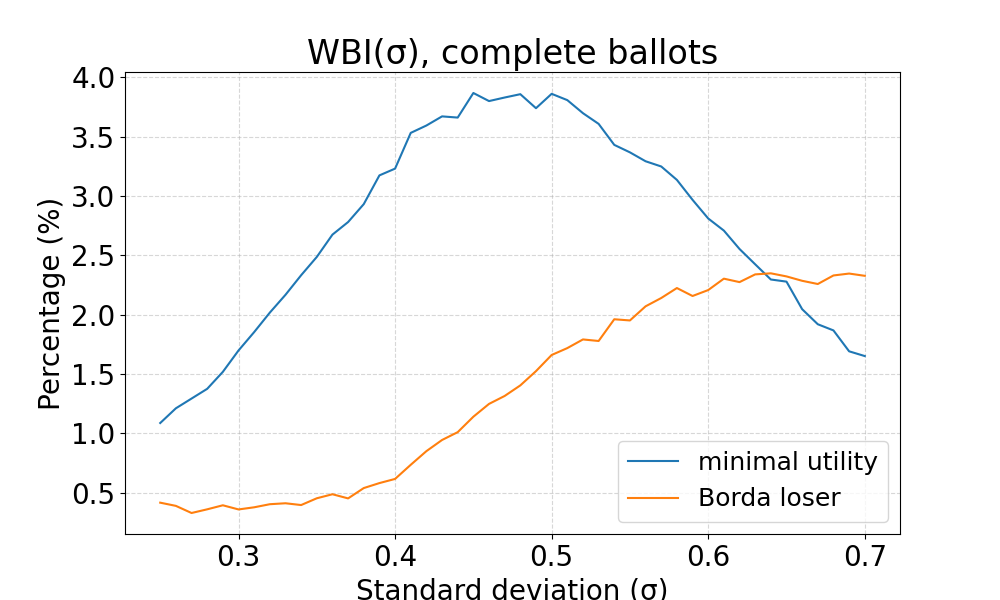}
    \end{minipage}

\begin{minipage}{0.485\textwidth}
        \centering
        \includegraphics[width=\linewidth]{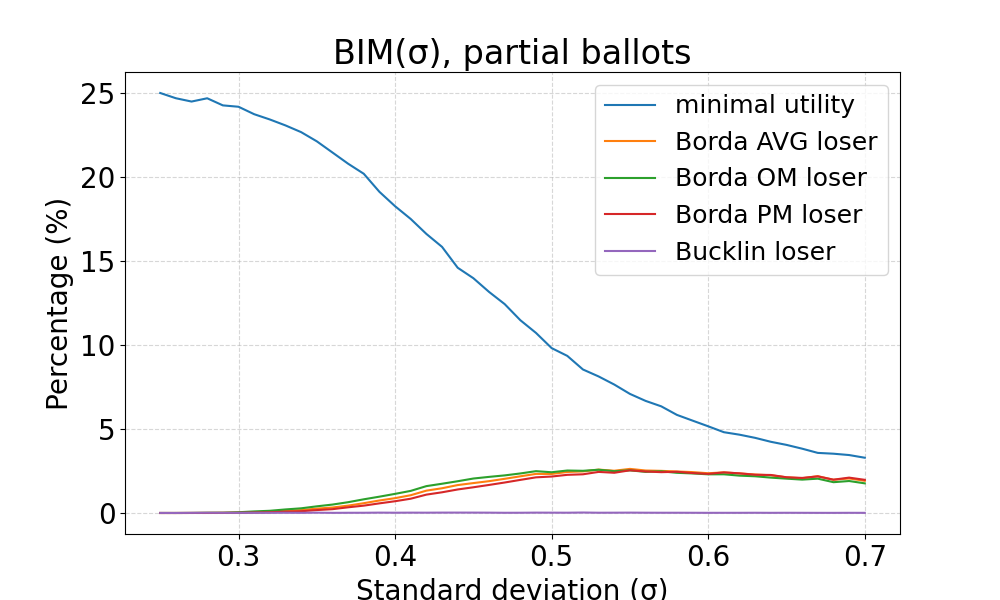}
    \end{minipage}
    \begin{minipage}{0.485\textwidth}
        \centering
        \includegraphics[width=\linewidth]{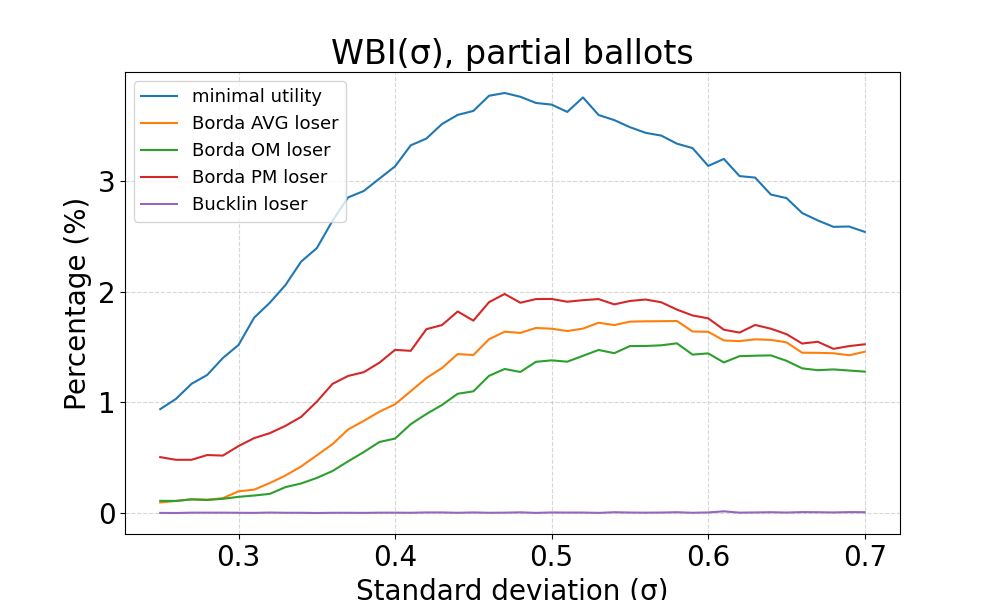}
    \end{minipage}
    \caption{Results under the 1D spatial models BIM$(\sigma)$ and WBI$(\sigma)$. The top two figures give results when every voter provides a complete ranking, and the bottom figures give results when a voter provides a complete ranking with probability 0.65.}
    \label{figure:1D_results}
\end{figure}

The figures demonstrate that different models produce quite different probabilities, sometimes in surprising ways. We explore why this is the case in Section \ref{section:model_differences}.
\section{Empirical Results}\label{section:empirical}

In this section we present results from the four real-world datasets described in Section \ref{section:models_datasets}. We remind the reader that in real-world elections we cannot analyze social utility, and thus we provide results for Borda and Bucklin losers. We also provide results for the candidate with the most last place votes only in the complete ballot setting. As in the previous section, for elections with partial ballots we exclude results for most last place votes, and use three Borda variants. We also recall that while most of the elections in the datasets have more than three candidates, we create a three-candidate preference profile from each by running the IRV algorithm until only three candidates remain. 

With the exception of a handful of elections from the state of Victoria in Australia, real-world elections have many partial ballots, especially when we begin with a large number of candidates and reduce to three using the IRV algorithm. To provide a comparison to our complete-preference simulation results, we complete each bullet vote to a ballot of length three using a proportional methodology as follows. Suppose 100 voters cast a bullet vote for $A$, 30 voters cast the ballot $A\succ B \succ C$, and 20 voters cast the ballot $A\succ C \succ B$. Then $100(30/50)$ of $A$'s bullet votes are completed to ballots of the form $A\succ B \succ C$ and $100(20/50)$ are completed to ballots of the form $A\succ C \succ B$. Any fractions are rounded to the nearest integer. While other ways of completing partial ballots are reasonable, we argue this approach preserves observed conditional preference information.

The results are shown in Table \ref{table:empirical_results}. The numbers in parentheses in the left column indicate the number of elections in each dataset. In relative terms, it is quite rare for IRV to elect a weakest candidate in real elections. The highest percentages occur in the \texttt{Scotland} dataset, which is likely a reflection of the multiwinner nature of these contests. The numbers significantly decrease when we move to \texttt{Scot\_condensed}.

\begin{table}
\centering
\begin{tabular}{l|ccc}
\multicolumn{4}{c}{\textbf{Ballots completed proportionally}}\\
\hline
\hline
& Borda loser & Bucklin loser & most last place\\
\hline
\texttt{Australia} (325) & 1 (0.31\%) & 11 (3.38\%) & 19 (5.85\%) \\
\texttt{Scotland} (1128)&17 (1.51\%) & 62 (5.50\%) & 101 (8.95\%)\\
\texttt{Scot\_condensed} (1128)&16 (1.42\%) &37 (3.28\%) &38 (3.37\%)\\
\texttt{USA} (401)& 1 (0.25\%) & 9 (2.24\%) & 11 (2.74\%)\\
\hline

\end{tabular}

\vspace{.1 in}

\begin{tabular}{l|cccc}
\multicolumn{5}{c}{\textbf{Actual ballots}}\\
\hline
\hline
& Borda AVG & Borda OM & Borda PM & Bucklin\\
\hline
\texttt{Australia} (325) &0 (0\%) &2 (0.62\%)& 0 (0\%)&0 (0\%) \\
\texttt{Scotland} (1128)& 4 (0.35\%) & 57 (5.05\%)& 10 (0.89\%)&26 (2.30\%)\\
\texttt{Scot\_condensed} (1128)&2 (0.18\%) &13 (1.15\%) &6 (0.53\%) &16 (1.42\%)\\
\texttt{USA} (401)&0 (0\%) &1 (0.25\%)& 1 (0.25\%)&2 (0.50\%) \\
\hline

\end{tabular}
\caption{The numbers and corresponding percentages of elections in which IRV chooses the weakest candidate in reduced three-candidate elections. The  numbers in parentheses in the leftmost column  denote the total number of elections in each dataset. Results are shown for the elections in which ballots were completed proportionally (Top) and with the actual elections (Bottom).}
\label{table:empirical_results}
\end{table}

An example of an election where IRV chooses both the Borda PM loser and the Bucklin loser is given below. This election was analyzed in \cite{MM23}.

\begin{example}\label{example:MN}
The 2021 Minneapolis ward 2 city council election contained the five candidates Tom Anderson, Yusra Arab, Guy Gaskin, Cam Gordon, and Robin Worlobah. Anderson and Gaskin were eliminated in the first two rounds, resulting in the preference profile shown in Table \ref{table:MN_profile}. The vote totals for Arab, Gordon, and Worlobah in this round are  3236, 2800, and 2879, respectively. Gordon is eliminated and Worlobah wins under IRV with 4056 votes to Arab's 4037.

Under Bucklin voting, no candidate receives an initial majority and thus we add first-rankings to second-rankings, resulting in scores of 5125, 5007, and 4812 for Arab, Gordon, and Worlobah, respectively. Thus, the IRV winner is the Bucklin loser in this election.

The Borda PM scores for the three candidates (in alphabetical order) are 8361, 7807, and 7691, and thus Worlobah is also the Borda PM loser. (It is straightforward to check the Borda AVG loser is Gordon and the Borda OM loser is Arab.)

\end{example}

As pointed out in \cite{MM23}, this election lacks a Condorcet winner since Arab beats Gordon head-to-head, Gordon beats Worlobah head-to-head, and Worlobah beats Arab head-to-head. Such an outcome is atypical for ranked-choice elections, and it is not surprising that IRV would choose the weakest candidate under some metrics in this situation. When there is no Condorcet winner it is difficult to identify a strongest and weakest candidate, and (as we see in this example) each candidate is likely to be ``weakest'' under some measure.

\begin{table}
\centering

\begin{tabular}{l|ccccccccc}
&908&756&1572&801&1177&822&1088&1299&492\\
\hline
1st&$A$&$A$&$A$&$G$&$G$&$G$&$W$&$W$&$W$\\
2nd&$G$&$W$&&$A$&$W$&&$A$&$G$&\\
3rd&$W$&$G$&&$W$&$A$&&$G$&$A$&\\
    
\end{tabular}
\caption{The preference profile for the 2021 Minneapolis ward 2 city council election, after reducing to the three candidates Arab ($A$), Gordon ($G$), and Worlobah ($W$).}
\label{table:MN_profile}

\end{table}

\section{Discussion}\label{section:discussion}

The preceding sections approached the same question from three complementary perspectives. Section \ref{section:analytical_results} provided exact probabilities under the classical IC and IAC models, giving \emph{a priori} baseline frequencies in electorates with no predetermined ideological structure. Section \ref{section:simulation_results} examined spatial models, allowing us to investigate the role of polarization and social utility. Section \ref{section:empirical} analyzed real-world election data, providing evidence about how often these outcomes occur in practice. Although these approaches differ substantially in their assumptions and methodology, they lead to a common conclusion: IRV can elect candidates who are weak under several reasonable notions of popularity, but such outcomes are generally quite rare. The primary exception occurs in highly polarized electorates, where the probability that IRV elects a weak candidate can become substantially larger. In this section, we discuss the mechanisms responsible for these patterns and explain why different notions of candidate weakness behave differently across our models. In particular, we explain the results in Sections \ref{section:simulation_results} and \ref{section:empirical}, comparing results among spatial models, and comparing simulated and empirical results. We also provide context for our results by comparing IRV to the methods of plurality and Condorcet.

\subsection{How does IRV choose a weakest candidate?}\label{why_weak}

Under our spatial models, instances in which IRV elects a weakest candidate tend to occur in conjunction with a \emph{center squeeze} outcome \cite{P13}, where a centrist Condorcet winner is eliminated in the first round.  Figure \ref{figure:1D_center_squeeze} illustrates two such examples under the BIM(0.5) model with complete ballots.

In the  image on the left,   $C$ is the Condorcet winner but receives the fewest first-place votes and is eliminated in the first round. Most of $C$’s support transfers to $A$, who narrowly defeats $B$. In this example, IRV does not elect the candidate with minimum social utility, $B$ (note that
$B$ lies farthest from the center), but does elect the Borda loser $A$. In this election $A$ is the Borda loser because  because $A$ and $C$ split the support of voters on the right, causing $B$ to have a strong first-place showing and to narrowly avoid having the smallest Borda score.

The  image on the right  shows a closely-related configuration in which the IRV winner, $B$, wins after the Condorcet winner $C$ is eliminated. However, because $A$ is slightly closer to the center, $B$ is also the candidate with the minimum social utility, while maintaining a small edge over $A$ in Borda scoring.

\begin{figure}
  \centering
  \begin{subfigure}[b]{0.48\textwidth}  
    \centering
    \includegraphics[width=\linewidth]{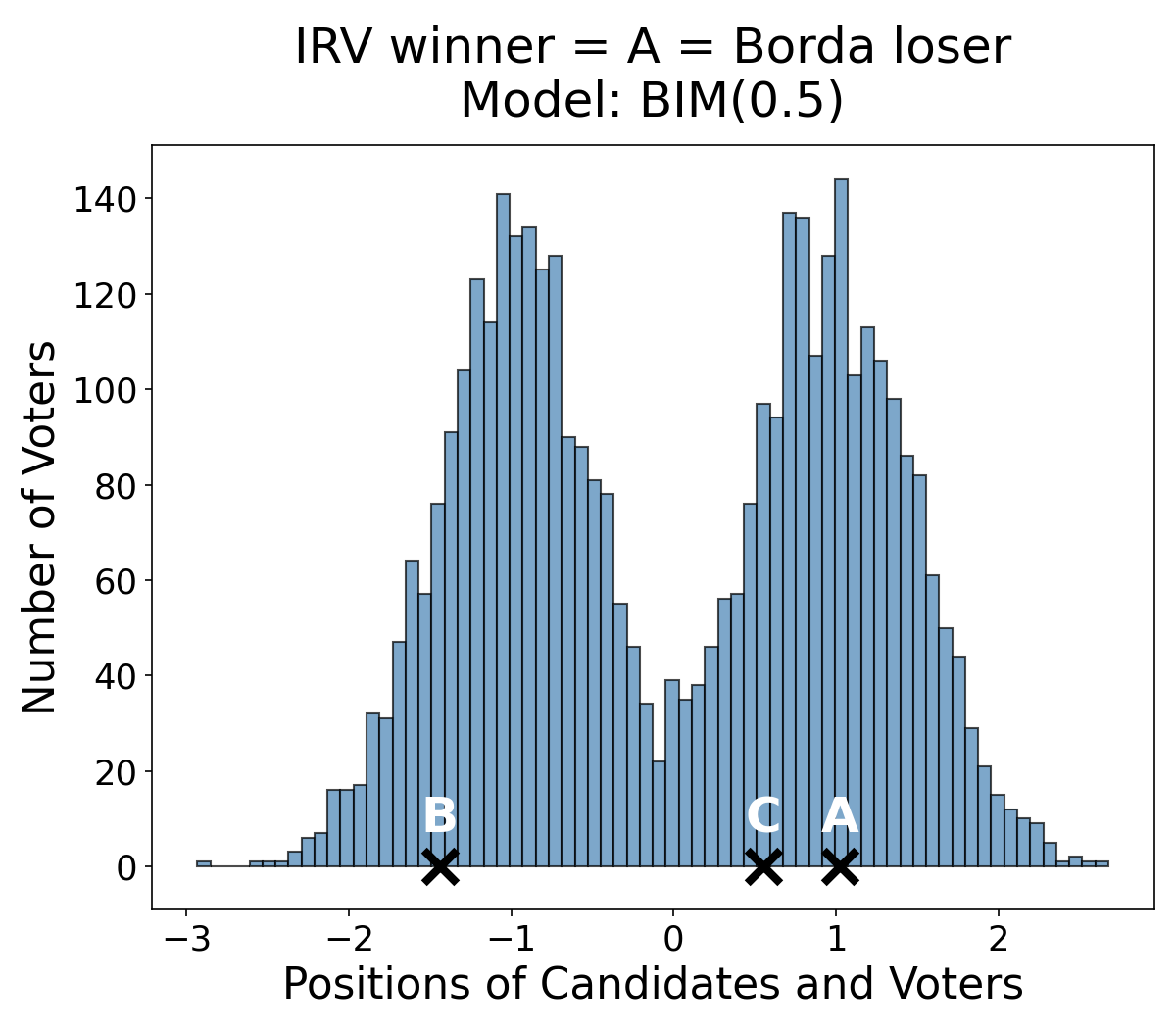}
  \end{subfigure}\hfill
  \begin{subfigure}[b]{0.48\textwidth}
    \centering
    \includegraphics[width=\linewidth]{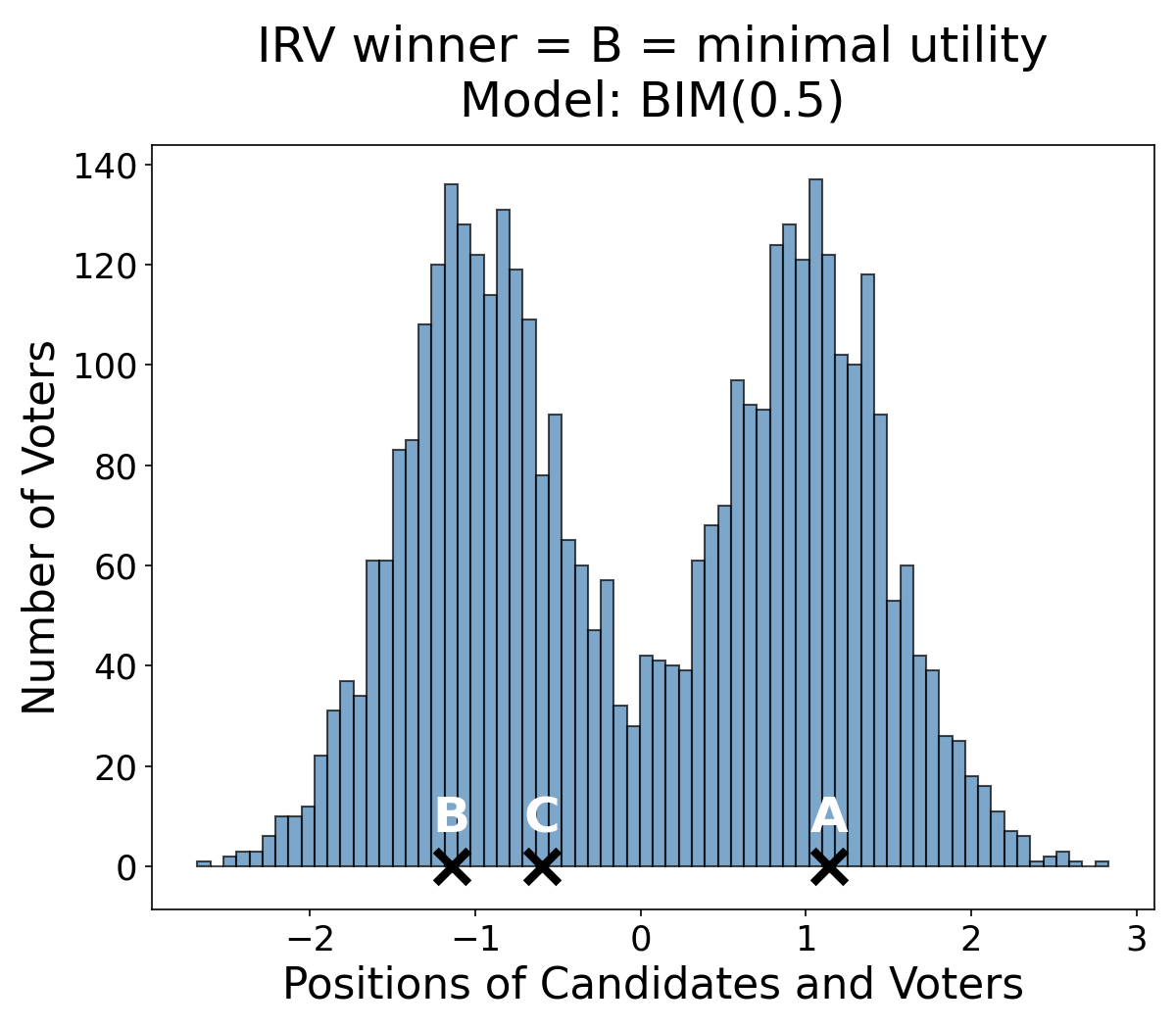}
  \end{subfigure}
  \caption{(Left) An election generated under the model BIM(0.5) where $A$ is the IRV winner and the Borda loser. (Right) An election generated under the model BIM(0.5) where the IRV winner $B$ minimizes social utility.}
  \label{figure:1D_center_squeeze}
\end{figure}

Figure \ref{figure:2D_center_squeeze} in Appendix 1 provides similar examples for 2D models. Although the  center squeeze phenomenon generally refers to 1D spatial models, it also applies in higher dimensions. 
Figure \ref{figure:2D_center_squeeze} (Left)  shows an election in which the candidate closest to the center, $C$, is eliminated in the first round and because $B$ is closer to the origin (the overall centroid of the electorate) than $A$, they are the IRV winner. However, $B$ is both the Borda loser and the candidate with minimum social utility. A similar dynamic occurs in Figure \ref{figure:2D_center_squeeze} (Right) where $B$ is the IRV winner and the Borda loser.


While these examples are ``typical'' in our simulations, there are, of course, other ways IRV can elect a weak candidate. For real-world context, we consider the elections in which the IRV candidate corresponds with one or more types of weak candidate in the datasets \texttt{Australia} and \texttt{USA} (actual ballots).  In some of these instances, the IRV winner does not coincide with the Condorcet winner, and hence these elections, while not mapped precisely to a spatial model, can be interpreted as a center squeeze. In other cases, IRV elects a weak candidate for different reasons.

As Table \ref{table:empirical_results} indicates, in the \texttt{Australia} dataset there are two elections in which IRV elects the Borda OM loser; otherwise IRV does not elect a weakest candidate. One of these two elections is the 2021 mayoral race in Hunters Hill, where a Condorcet winner exists but is not elected by IRV. This is best understood as a center squeeze outcome. The other election is a 2015 state legislative assembly election in Lismore, where the IRV and Condorcet winners agree and are both the Borda OM loser; this does not correspond to a center squeeze outcome.

In the dataset \texttt{USA}, the numbers $1, 1$ and $2$ in Table  \ref{table:empirical_results}, indicating the number of elections where the IRV winner is the Borda OM, Borda PM, and Bucklin loser respectively, correspond to only two different elections. The first is analyzed in Example \ref{example:MN}. This election does not contain a Condorcet winner, in which case  it is not surprising for any voting method to elect a weak candidate under some measure. The other election is a 2021 multi-winner city council election from Moab, UT, where  IRV does not elect the Condorcet winner  but elects the Borda OM and Bucklin loser. This is best understood as a center squeeze outcome.

These examples show that in real-world elections, IRV  may elect  weak candidate in a number of different ways. They also suggest that empirical data does not map precisely onto any of these spatial models, and that voter consensus in real elections about strong or weak candidates often cannot be neatly captured by spatial models such as ours.


\subsection{Differences among the spatial models}\label{section:model_differences}

The differences in probabilities among the different spatial models raise many questions, both about the dynamics of the IRV algorithm as well as the identity of the different weak losers in these different settings. It is not feasible to explain the behavior of every aspect of every model, so to keep the narrative streamlined, we focus here on two of the most striking features of Figure \ref{figure:1D_results}.

The first concerns some of the differences between Figure \ref{figure:1D_results} (Left) and (Right). Why are the Borda loser probabilities so large for very polarized electorates under the BIM$(\sigma)$ model, but not for the WBI$(\sigma)$ model? Furthermore, why, as  $\sigma$ increases, do the Borda probabilities decrease  under BIM, but increase under WBI? We cannot provide definitive answers, but based on an analysis of candidate placements in many generated elections, we offer a partial explanation.  Under the BIM model, IRV generally elects the Borda loser  when the candidate placement is as shown in Figure \ref{figure:1D_center_squeeze} (Left), and the situation is as described in Section \ref{why_weak}.  Under the WBI model however,  IRV  generally elects the Borda loser when the candidate placement is as shown in Figure \ref{figure:1D_model_difference} (Left). In this scenario, two candidates are placed firmly in the left camp, with the IRV winner $A$ to the left of the Condorcet winner, $C$, with the third candidate, $B$,  in the interval $(0, 0.5)$. That is, for IRV to elect the Borda loser under a WBI model,  there must be some voters in the center so that $B$ receives a larger Borda score than $A$.  If $\sigma$ is small, then there are almost no voters in the interval $(0, 0.5)$ and hence the probability shown in Figure \ref{figure:1D_results} is very small for WBI$(\sigma)$. As $\sigma$ increases, the probability of generating such a scenario increases. 

The second question we consider is why are the partial ballot Borda loser probabilities so much lower than the complete ballot probabilities for the BIM models (compare Figure \ref{figure:1D_results} (Top and Bottom Left))? We cannot know for certain but, as with the previous question, we hazard an answer for the Borda AVG candidate based on  an analysis of candidate placements in many generated elections. When partial ballots are allowed, it appears IRV elects the Borda AVG loser when candidates are as shown in Figure \ref{figure:1D_model_difference} (Right). Two criteria appear to be necessary:  (1) two candidates, ($A$ and $C$  in Figure \ref{figure:1D_model_difference}) must be placed very closely together on one side of the spectrum, near 1 or $-1$; and (2) the remaining candidate ($B$) must be placed on the opposite side of the spectrum and even further from the center. In fact, in most of the generated elections in which IRV elects the Borda AVG loser, the candidate placement is even more extreme than that  in Figure \ref{figure:1D_model_difference}.  We do not know why these exact conditions allow for IRV to select the Borda AVG loser under random ballot truncation, but the narrowness of these conditions explains why the probabilities in the bottom left of Figure \ref{figure:1D_results} are so small.

\begin{figure}
  \centering
  \begin{subfigure}[b]{0.48\textwidth}  
    \centering
    \includegraphics[width=\linewidth]{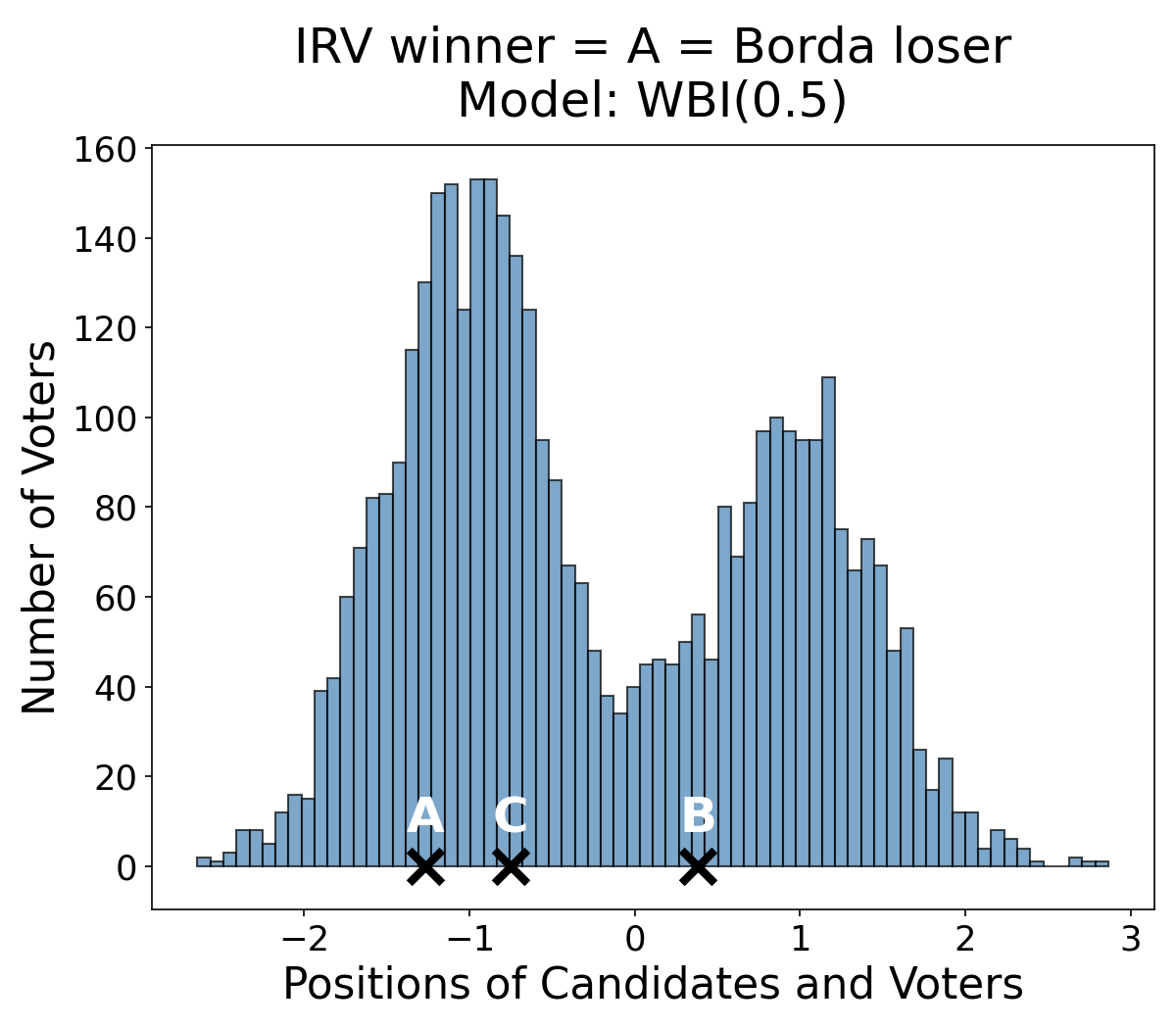}
  \end{subfigure}\hfill
  \begin{subfigure}[b]{0.48\textwidth}
    \centering
    
    \includegraphics[width=\linewidth]{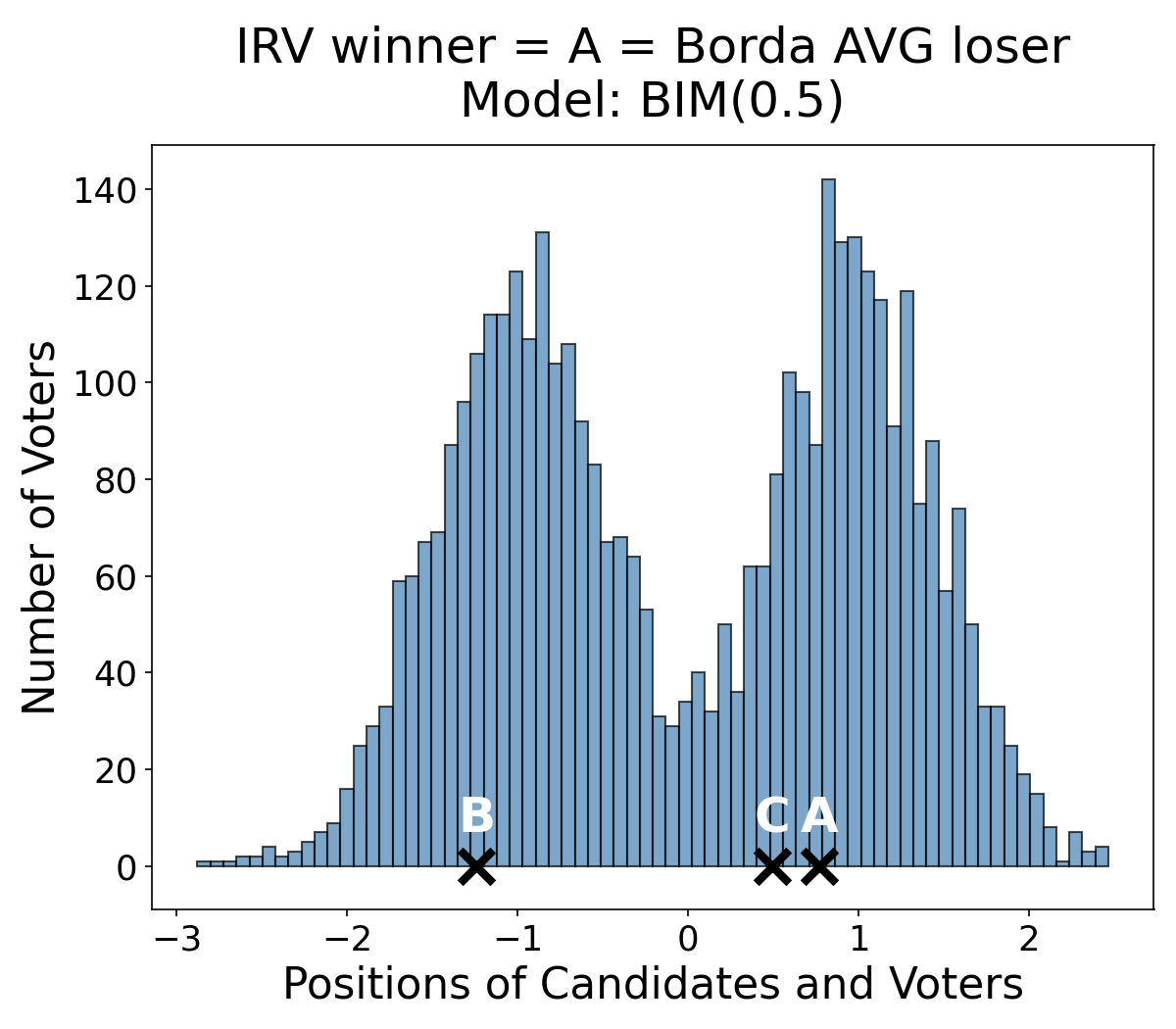}
  \end{subfigure}
  \caption{(Left) An election generated under the model WBI(0.5) where $A$ is the IRV winner and the Borda loser. (Right) An election generated under the model BIM(0.5) with partial ballots where the IRV winner $A$ is the Borda AVG loser.}
  \label{figure:1D_model_difference}
\end{figure}

\subsection{Differences between  Models and Data}
The broad range of models and sizable differences in probabilities even within the same model make it 
 difficult to make universal comparisons between  analytical, simulated, and empirical results. However, we can venture some general observations. The first thing to note is that the theoretical models do not consistently predict higher likelihoods of weak candidates being elected then what occurs in the data. This makes the analysis of IRV   different from other analyses. 

 In previous studies  of IRV, the  probabilities of anomalous behavior based on \emph{a priori} models such as IAC or IC, as well as different spatial models, tend to be much larger than those observed in real-world data. For example, theoretical models predict a much higher probability that an election demonstrates an upward monotonicity paradox than what occurs in actual elections \cite{GSM23, M17}. This appears not to be  true for the likelihood of IRV electing weak candidates, as least when partial ballots are allowed.

In the complete ballot case, the probability that IRV elects the Borda loser is generally lower  in the real-world data than under IAC, IC,  or the spatial models. However, the  probability that IRV selects the Bucklin loser or the candidate with the most last place votes is sometimes higher in the real elections than in simulated elections. For example, the probability  of electing the Bucklin loser in the \texttt{Scotland} dataset is higher than the same probability under IC, IAC, and several of the spatial models.
 The only simulated elections that return higher Bucklin loser probabilities occur under the model UNI $\times$ BIM($\sigma$) for small values of $\sigma$. The situation is similar for the candidate with the most last place votes. As mentioned previously, the larger probabilities for \texttt{Scotland} most likely reflect that most of these elections are not single-winner and have been reduced to three candidates.  In the other datasets the probabilities are almost always lower than what other models predict.

In the partial ballot case, the probabilities that IRV selects the various Borda losers or the Bucklin loser are similar between real-world and simulated elections. These probabilities are quite low, with the only values above 5\% occurring under the UNI $\times$ BIM($\sigma$) model for small $\sigma$. These results reinforce our previous analysis that in real-world elections it is unlikely for a Borda or Bucklin loser to be elected by IRV. We would expect to see such outcomes if an electorate is extremely polarized in some way, or if the election is very ``close'' (e.g. there is no Condorcet winner).

\subsection{Comparison to other voting methods}

In this section, we place our IRV results in a broader context by comparing the performance of IRV to two other voting methods: plurality and Condorcet. Under \text{plurality}, the candidate with the most first-place votes is the election winner. 
Under a \textbf{Condorcet method}, the Condorcet winner is  declared the  winner of the election, when such a candidate exists. When there is no Condorcet winner, an alternative criterion is used.  For our analysis we  use the   minimax method, which selects the candidate whose worst head-to-head loss is minimal when there is no Condorcet winner. 
However, a Condorcet winner exists in more than 99\% of both our simulated and real elections, and thus the choice of Condorcet method is largely immaterial. 

Figures \ref{figure:method_comparison_utility} and \ref{figure:method_comparison_Borda} in Appendix 1 show how IRV compares with  plurality and Condorcet with respect to the candidate of minimal utility and the Borda loser under two-dimensional models involving BIM($\sigma$). Figure  \ref{figure:method_comparison_utility} suggests that the plurality method is much more likely to elect the candidate of minimal utility. Condorcet generally performs slightly better than IRV; the difference is generally not substantial except when the electorate is highly polarized. The introduction of partial ballots causes both Condorcet and IRV to choose the candidate of minimal utility with higher frequency, and the methods tend to choose this candidate at approximately the same rate.

The situation is different for the Borda loser since the Condorcet winner can never be the Borda loser or Borda AVG under complete or partial ballots respectively; see \cite{FB25} for a proof that the Condorcet winner can never be the Borda AVG loser.  Figure \ref{figure:method_comparison_Borda} shows that Condorcet clearly outperforms the other two methods, even when measuring the probability of selecting the Borda OM loser. The figure also shows, somewhat surprisingly, that plurality performs well relative to IRV under models like BIM($\sigma$). This is likely because in a bimodal electorate that is highly polarized, in most instances, two candidates are firmly in one camp and the third candidate in the other. The latter is the plurality winner since they do not split the first-place votes of their supporters with any other candidate, and this strong first-place showing prevents them from being the Borda loser. Even though plurality sometimes outperforms IRV for this metric, plurality performs far worse than IRV with respect to the candidate with minimum social utility, and with respect to some additional criteria beyond the scope of this article. For example, plurality elects the Condorcet loser at a much higher rate than any  commonly-studied ranked voting rule.

While Condorcet dramatically outperforms IRV with respect to the  Borda loser, we note that Condorcet methods perform worse than IRV for metrics not considered by this article. For example, while IRV never elects the plurality loser,  Condorcet methods frequently do so under our models. Indeed, if an election contains three candidates and a Condorcet winner exists, IRV and Condorcet disagree if and only if the Condorcet winner is the plurality loser. (If this were not the case, then the Condorcet winner would win in the second round of IRV.)  Of course, electing a plurality loser may be desirable in some cases, but it is not clear that an electorate would frequently tolerate the election of the candidate ``with the fewest votes.''

\section{Conclusion}\label{section:conclusion}

Can ranked-choice voting (instant runoff voting) elect the weakest candidate? The answer is yes, for several natural notions of \emph{weakest}, in the three-candidate elections we study. This observation is not unique to IRV—any reasonable voting rule can fail under some criteria—but it establishes that such outcomes are neither purely theoretical nor confined to contrived examples. Across our models and empirical data, the frequency with which IRV elects a weak candidate is generally low, rising to non-negligible levels primarily in highly polarized electorates where voter preferences have a spatial structure.

There are several natural directions for future work. First, we could examine elections with four or more candidates. Since real elections rarely feature more than three viable candidates, we expect the frequency with which IRV elects the weakest candidate (under any measure) to remain small, though this remains to be quantified. Second, we could focus on a single notion of weakness and compare voting methods by the probability with which they fail under that criterion. For example, it would be interesting to identify which commonly studied voting rules minimize the probability of electing the candidate of minimal social utility. It would also be interesting to confirm that the results are robust based on different measures of social utility. Finally, we could examine how varying degrees of ballot truncation or strategic voting affect these results.

Overall, our work contributes to the voting literature by clarifying when and how IRV can fail under intuitive notions of weakness, and by identifying the electoral conditions under which such failures are likely to occur.\\


\section*{Appendix 1: Figures}
\begin{figure}[htbp]
    \centering

    \begin{minipage}{0.48\textwidth}
        \centering
        \includegraphics[width=\linewidth]{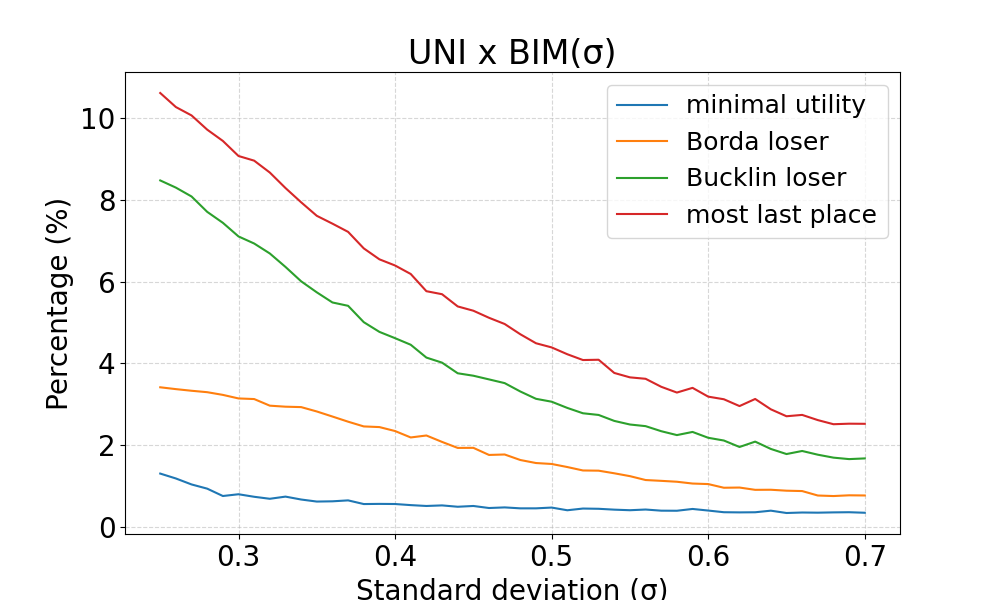}
    \end{minipage}
    \hspace{0.02\textwidth}
    \begin{minipage}{0.48\textwidth}
        \centering
        \includegraphics[width=\linewidth]{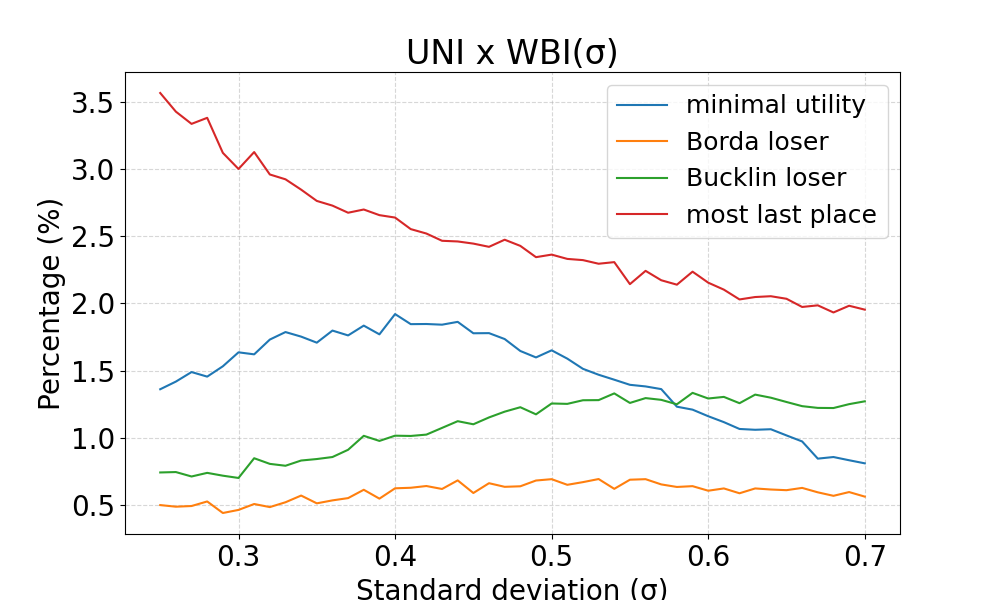}
    \end{minipage}

    \vspace{2mm} 

    \begin{minipage}{0.48\textwidth}
        \centering
        \includegraphics[width=\linewidth]{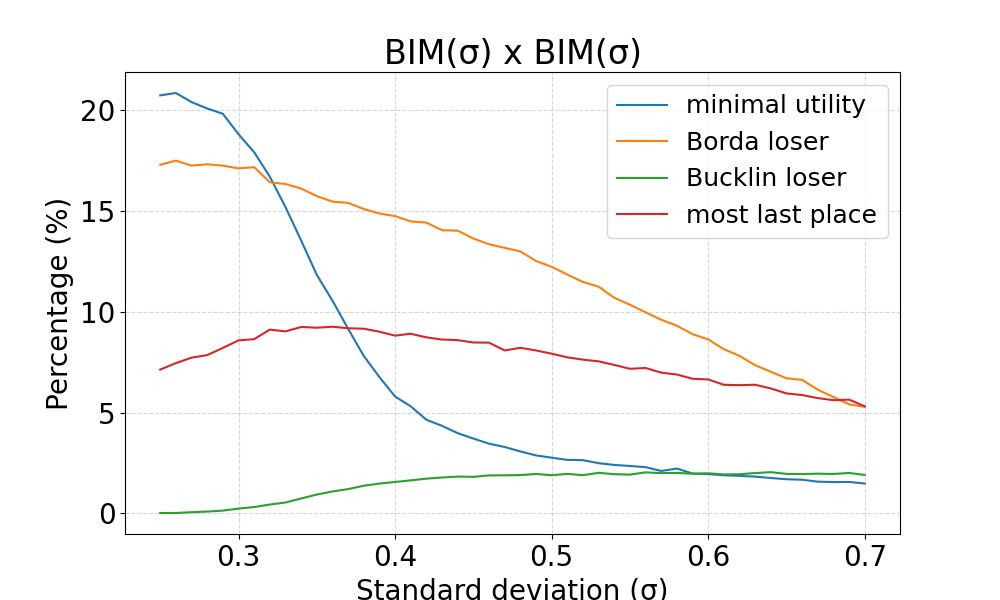}
    \end{minipage}
    \hspace{0.02\textwidth}
    \begin{minipage}{0.48\textwidth}
        \centering
        \includegraphics[width=\linewidth]{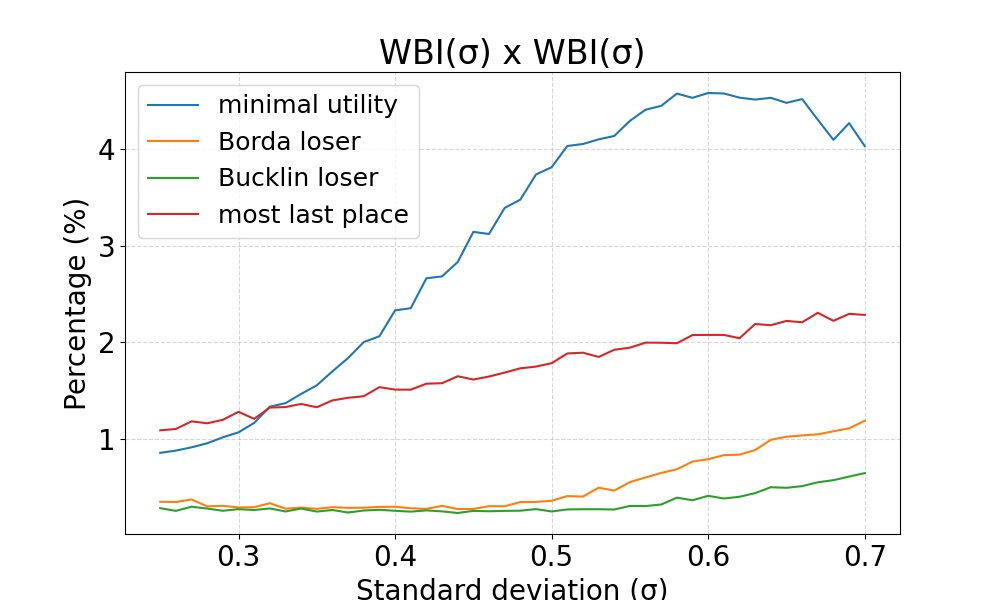}
    \end{minipage}

    \vspace{2mm} 

    \begin{minipage}{0.5\textwidth}
        \centering
        \includegraphics[width=\linewidth]{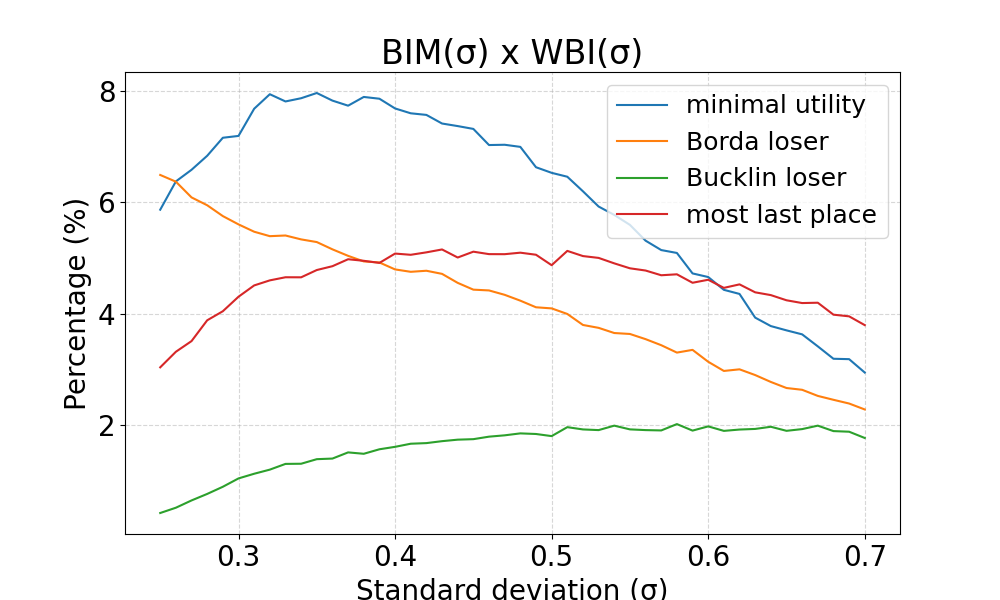}
    \end{minipage}

    \caption{The percentage of elections in which IRV elects various kinds of weakest candidate under 2D spatial models with complete ballots.}
    \label{figure:2D_complete_ballots}
\end{figure}

\begin{figure}[htbp]
    \centering

    \begin{minipage}{0.48\textwidth}
        \centering
        \includegraphics[width=\linewidth]{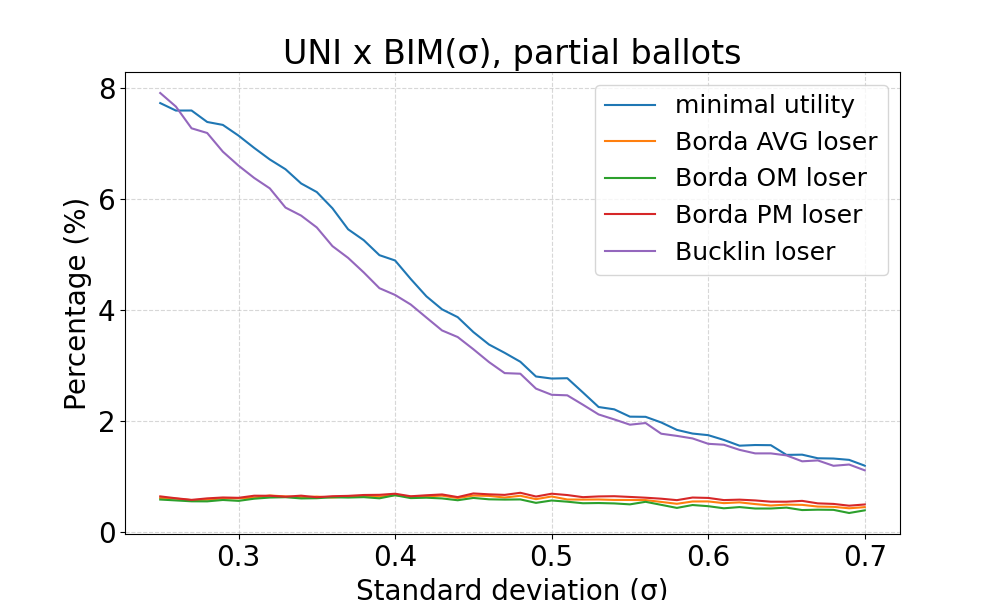}
    \end{minipage}
    \hspace{0.02\textwidth}
    \begin{minipage}{0.48\textwidth}
        \centering
        \includegraphics[width=\linewidth]{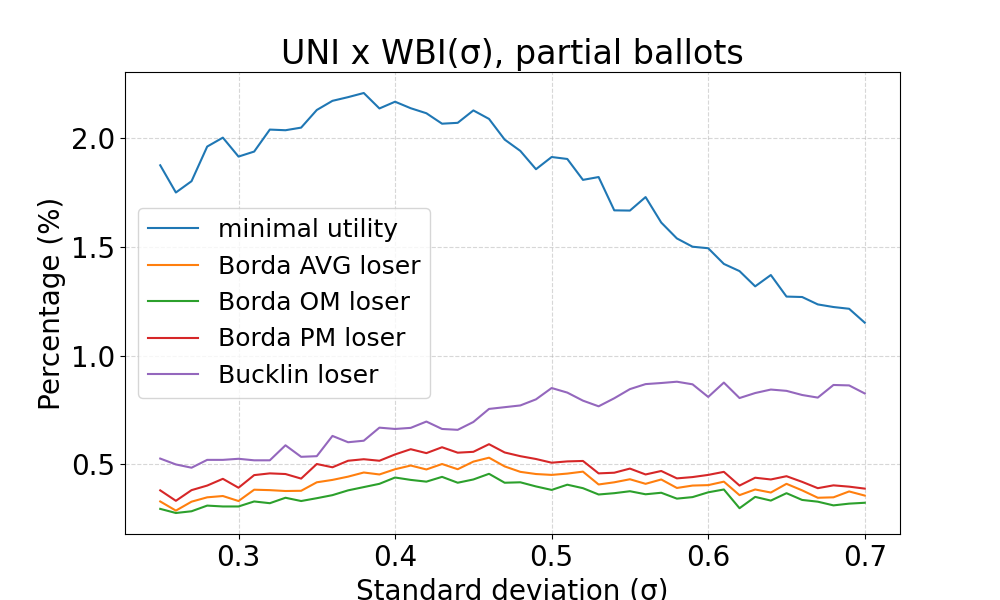}
    \end{minipage}

    \vspace{2mm} 

    \begin{minipage}{0.48\textwidth}
        \centering
        \includegraphics[width=\linewidth]{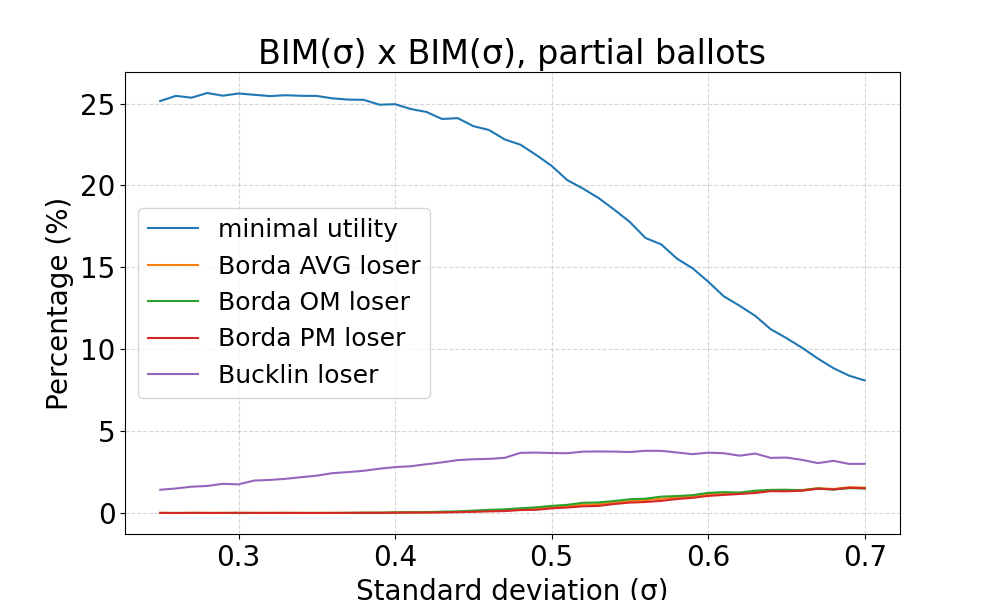}
    \end{minipage}
    \hspace{0.02\textwidth}
    \begin{minipage}{0.48\textwidth}
        \centering
        \includegraphics[width=\linewidth]{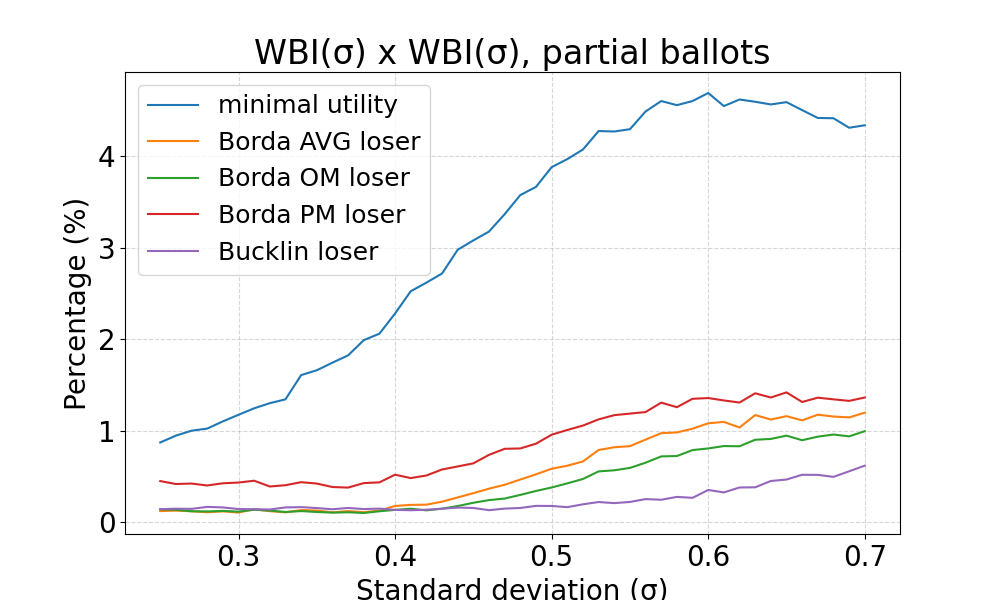}
    \end{minipage}

    \vspace{2mm} 

    \begin{minipage}{0.5\textwidth}
        \centering
        \includegraphics[width=\linewidth]{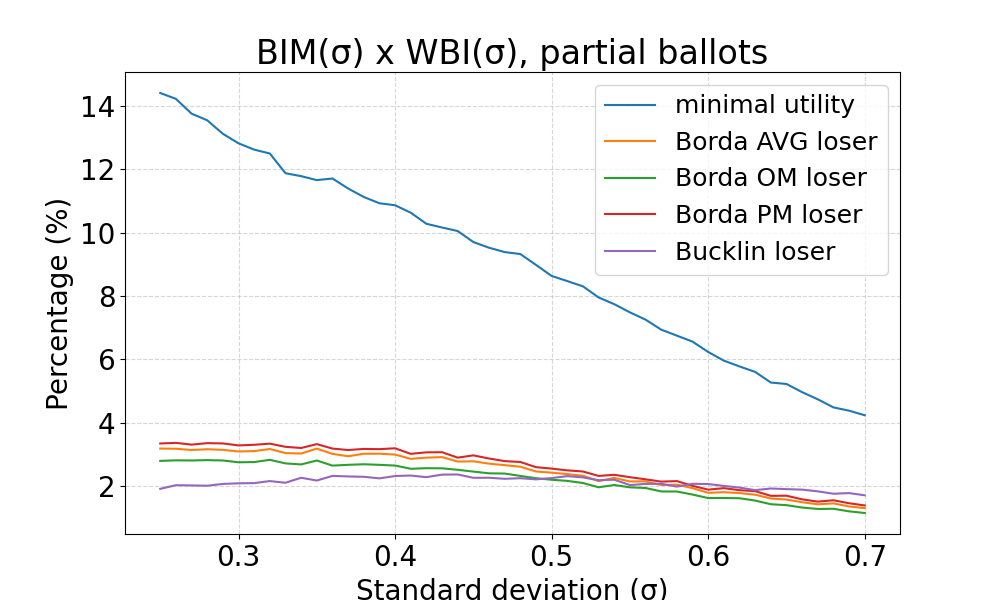}
    \end{minipage}

    \caption{The percentage of elections in which IRV elects various kinds of weakest candidate under 2D spatial models with partial ballots.}
    \label{figure:2D_partial_ballots}
\end{figure}

\begin{figure}[htbp]
  \centering
\includegraphics[width=130mm]{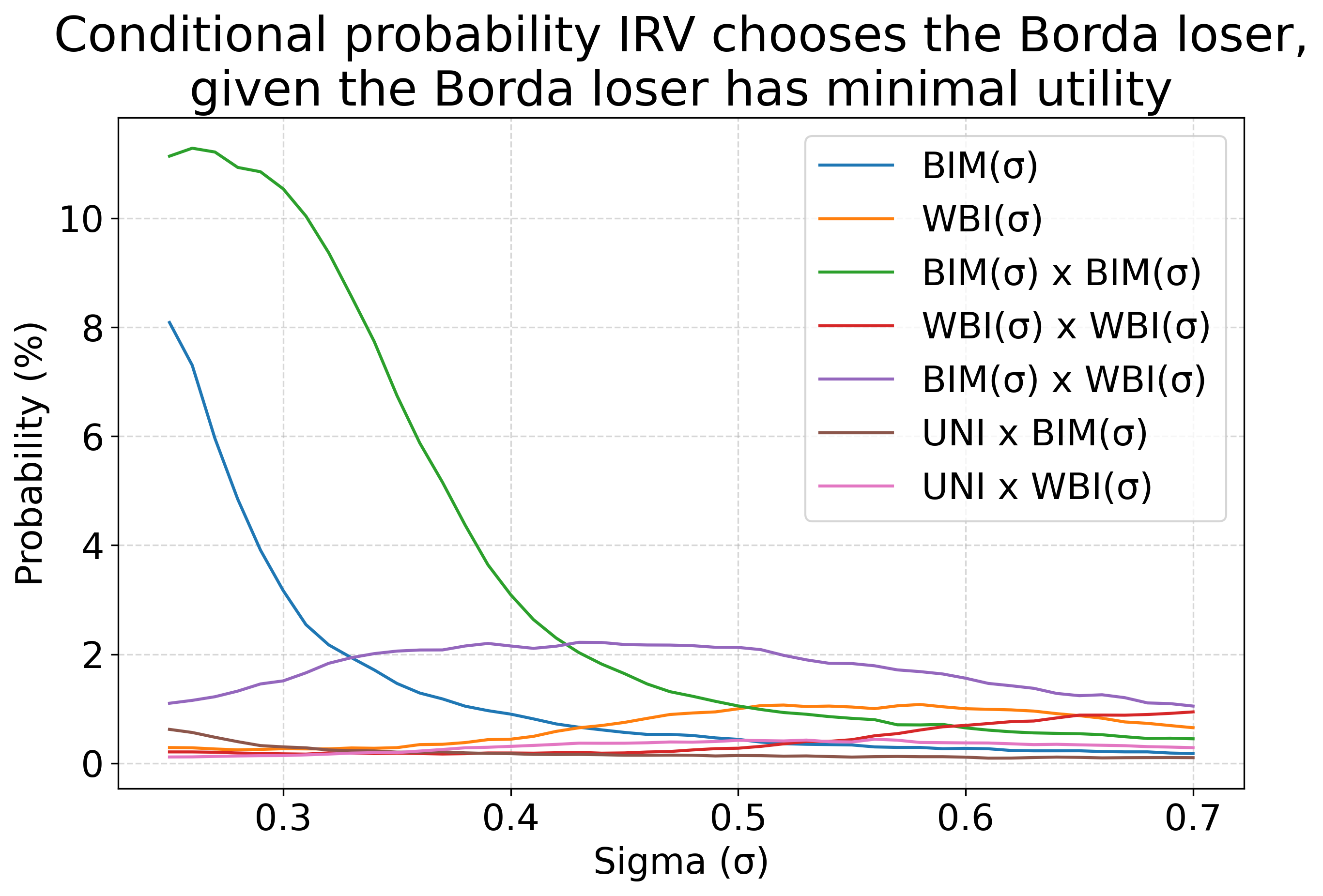}

\vspace{.1 in}

\includegraphics[width=130mm]{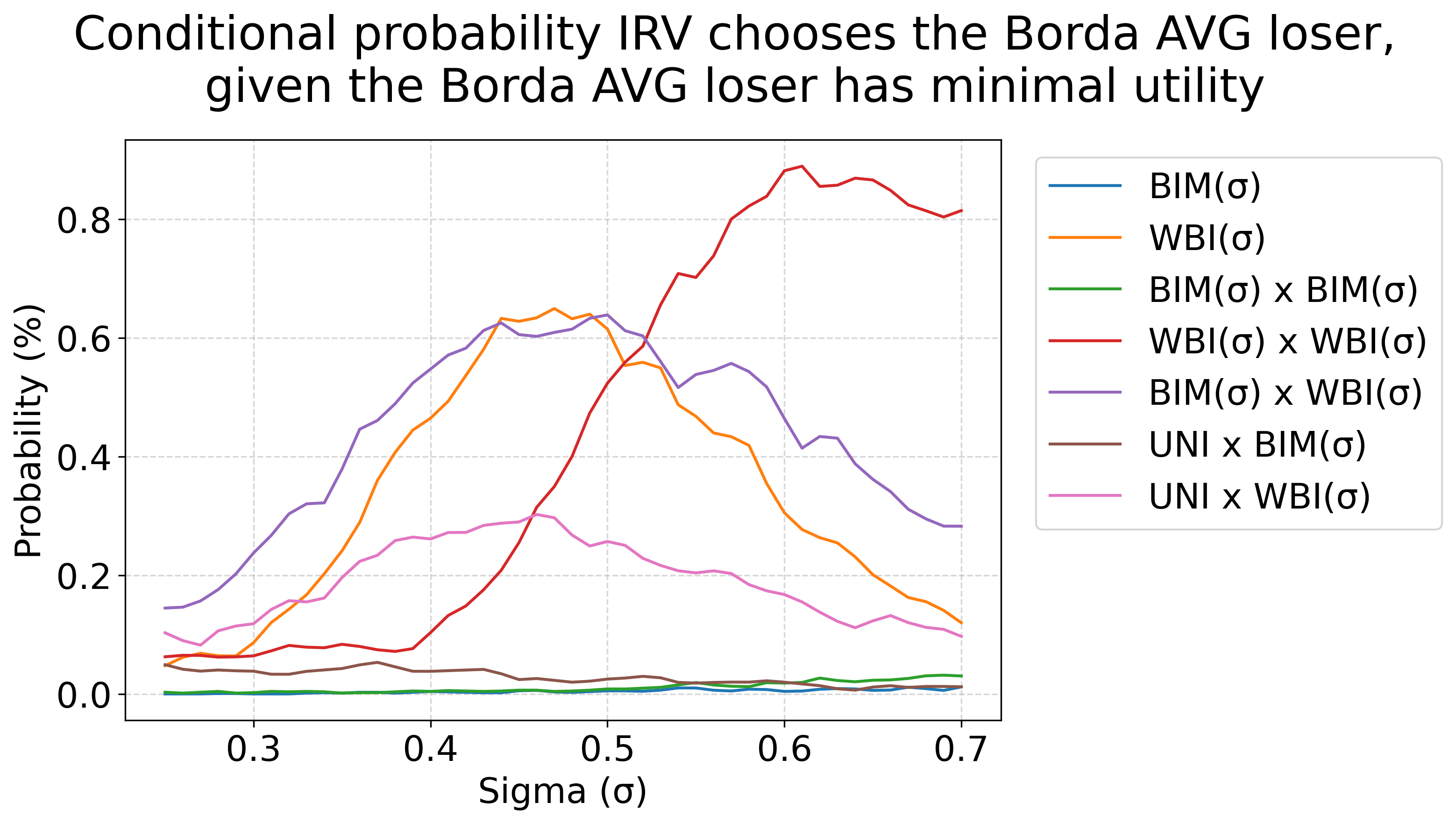}
\caption{(Top) The probability (expressed as a percentage) that IRV elects the Borda loser, assuming the Borda loser is the candidate of minimal social utility with complete  ballots; (Bottom) the probability that IRV elects the  Borda AVG loser, with partial ballots.}
\label{figure:cond_prob_Borda_is_minimal_utility}

  \end{figure}

\begin{figure}[htbp]
    \centering
\begin{minipage}{0.485\textwidth}
        \centering
        \includegraphics[width=\linewidth]{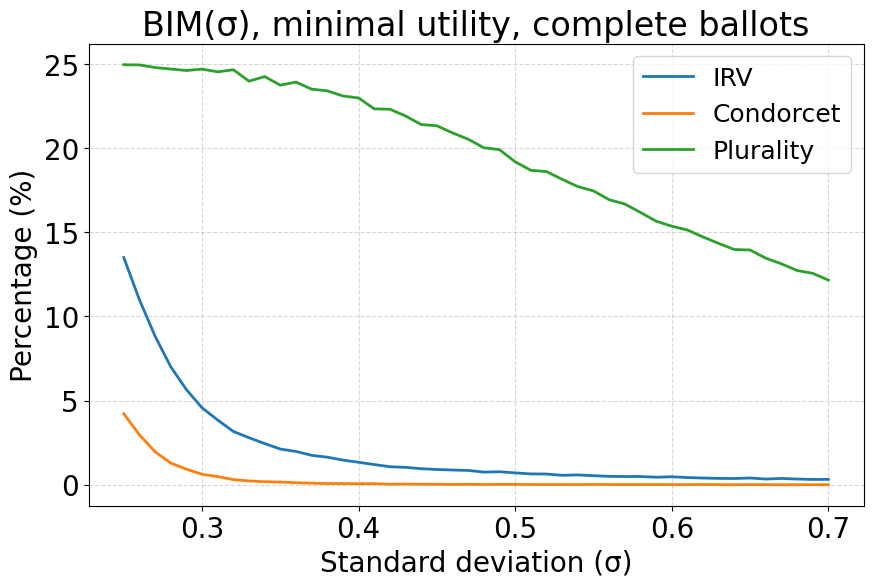}
    \end{minipage}
    \begin{minipage}{0.485\textwidth}
        \centering
        \includegraphics[width=\linewidth]{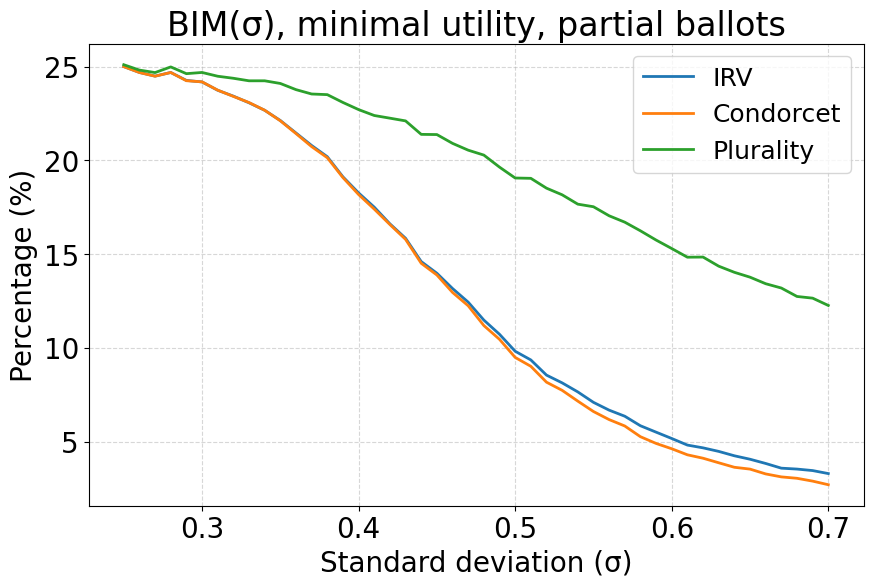}
    \end{minipage}

    \vspace{.1 in}
    \begin{minipage}{0.485\textwidth}
        \centering
        \includegraphics[width=\linewidth]{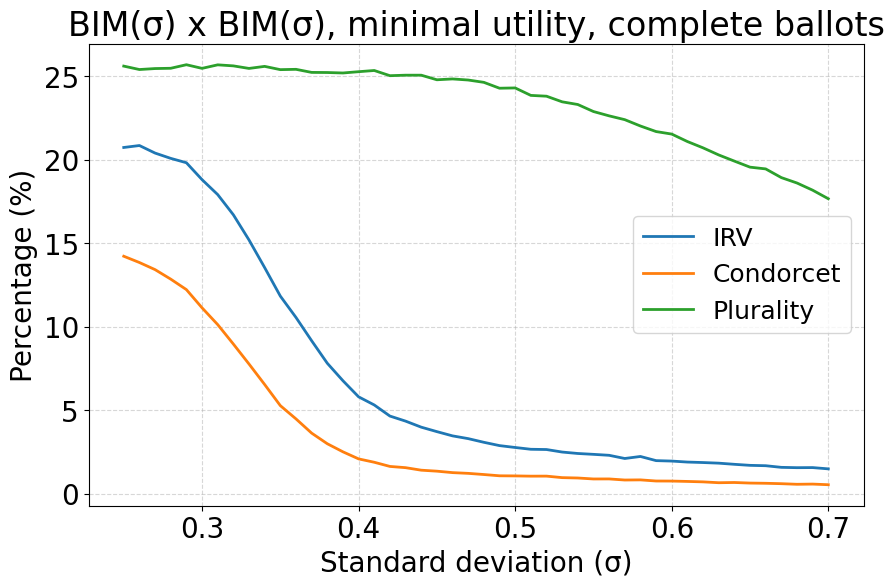}
    \end{minipage}
    \begin{minipage}{0.485\textwidth}
        \centering
        \includegraphics[width=\linewidth]{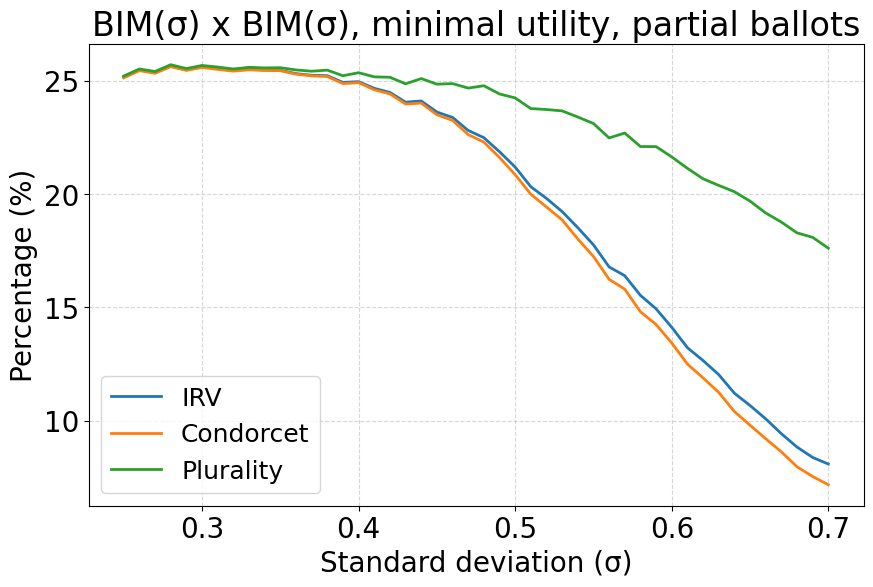}
    \end{minipage}
    \vspace{.1 in}
\begin{minipage}{0.485\textwidth}
        \centering
        \includegraphics[width=\linewidth]{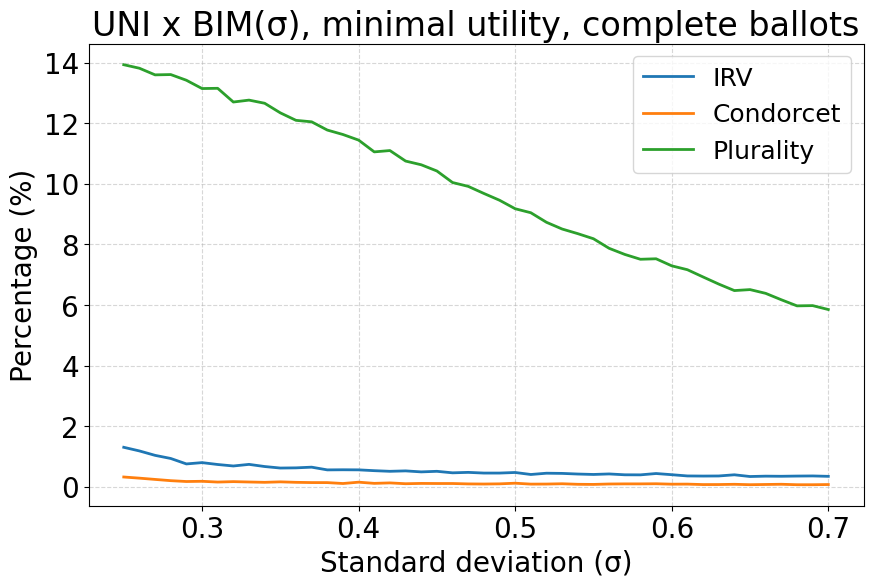}
    \end{minipage}
    \begin{minipage}{0.485\textwidth}
        \centering
        \includegraphics[width=\linewidth]{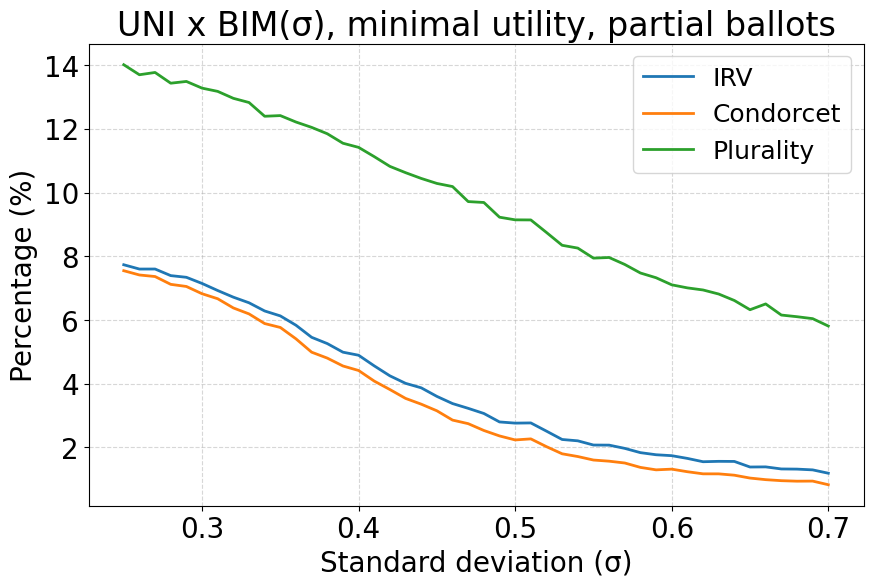}
    \end{minipage}

    \caption{A comparison of the performance of the voting rules IRV, plurality, and Condorcet for the measure of minimal social utility.}
    \label{figure:method_comparison_utility}
\end{figure}

\begin{figure}[htbp]
    \centering
\begin{minipage}{0.485\textwidth}
        \centering
        \includegraphics[width=\linewidth]{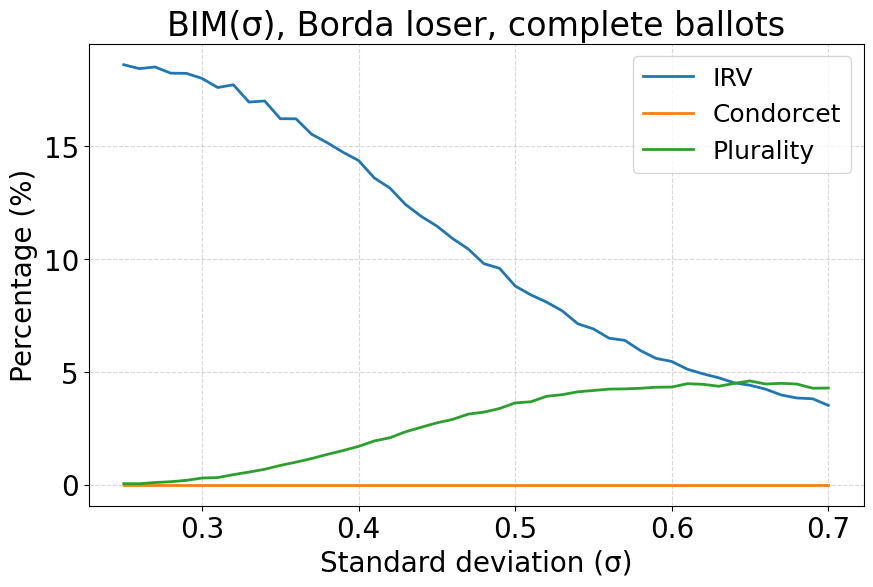}
    \end{minipage}
    \begin{minipage}{0.485\textwidth}
        \centering
        \includegraphics[width=\linewidth]{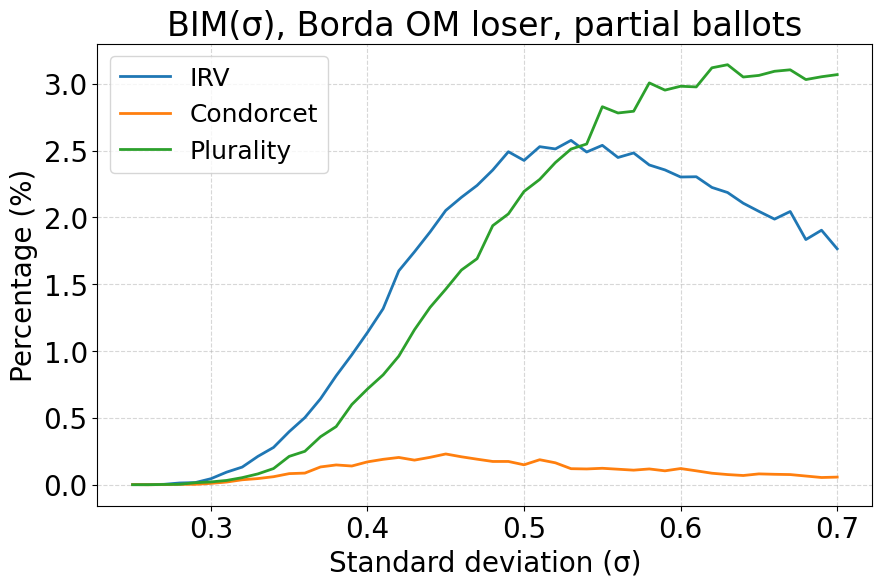}
    \end{minipage}

    \vspace{.1 in}
    \begin{minipage}{0.485\textwidth}
        \centering
        \includegraphics[width=\linewidth]{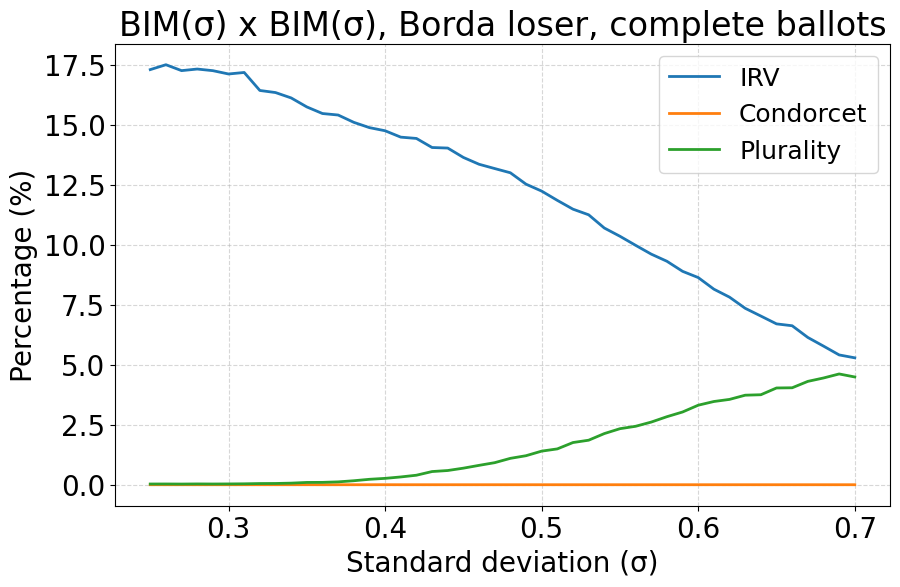}
    \end{minipage}
    \begin{minipage}{0.485\textwidth}
        \centering
        \includegraphics[width=\linewidth]{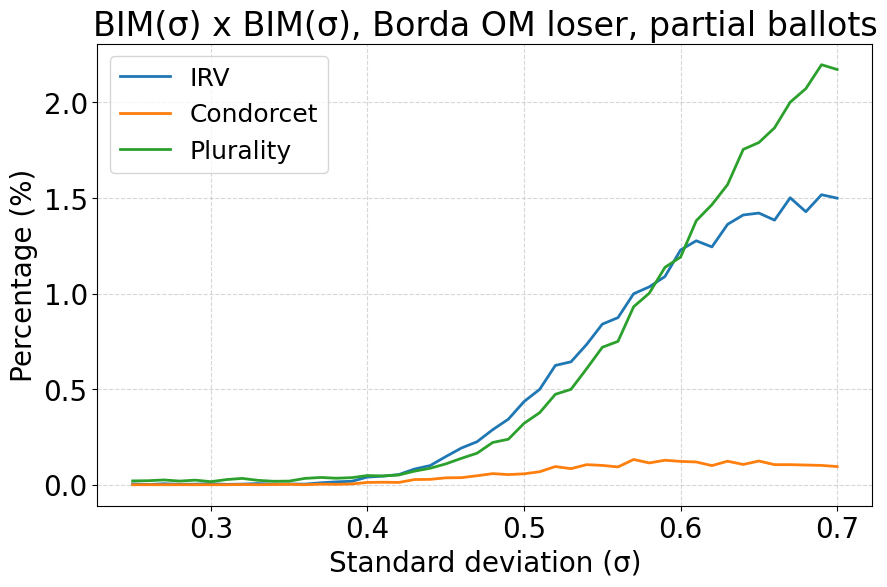}
    \end{minipage}
    \vspace{.1 in}
\begin{minipage}{0.485\textwidth}
        \centering
        \includegraphics[width=\linewidth]{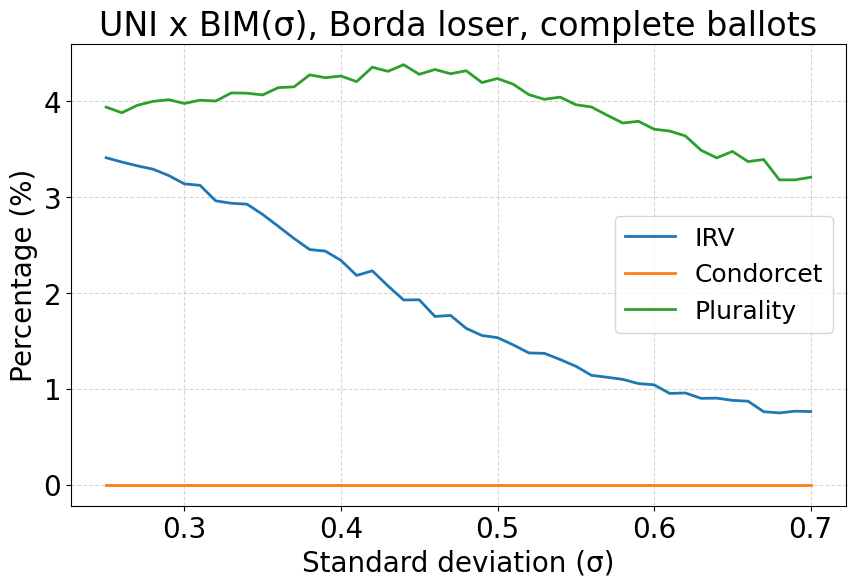}
    \end{minipage}
    \begin{minipage}{0.485\textwidth}
        \centering
        \includegraphics[width=\linewidth]{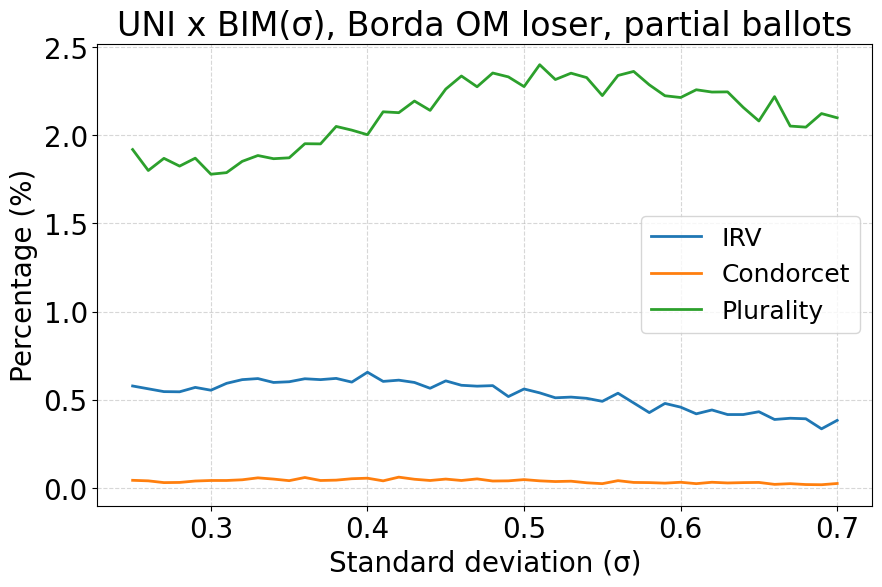}
    \end{minipage}

    \caption{A comparison of the performance of the voting rules IRV, plurality, and Condorcet for the measure of Borda loser for complete ballots and Borda OM loser for partial ballots.}
    \label{figure:method_comparison_Borda}
\end{figure}

\begin{figure}[htbp]
  \centering
  \begin{subfigure}[b]{0.48\textwidth}
    \centering
    \includegraphics[width=\linewidth, height=65mm]{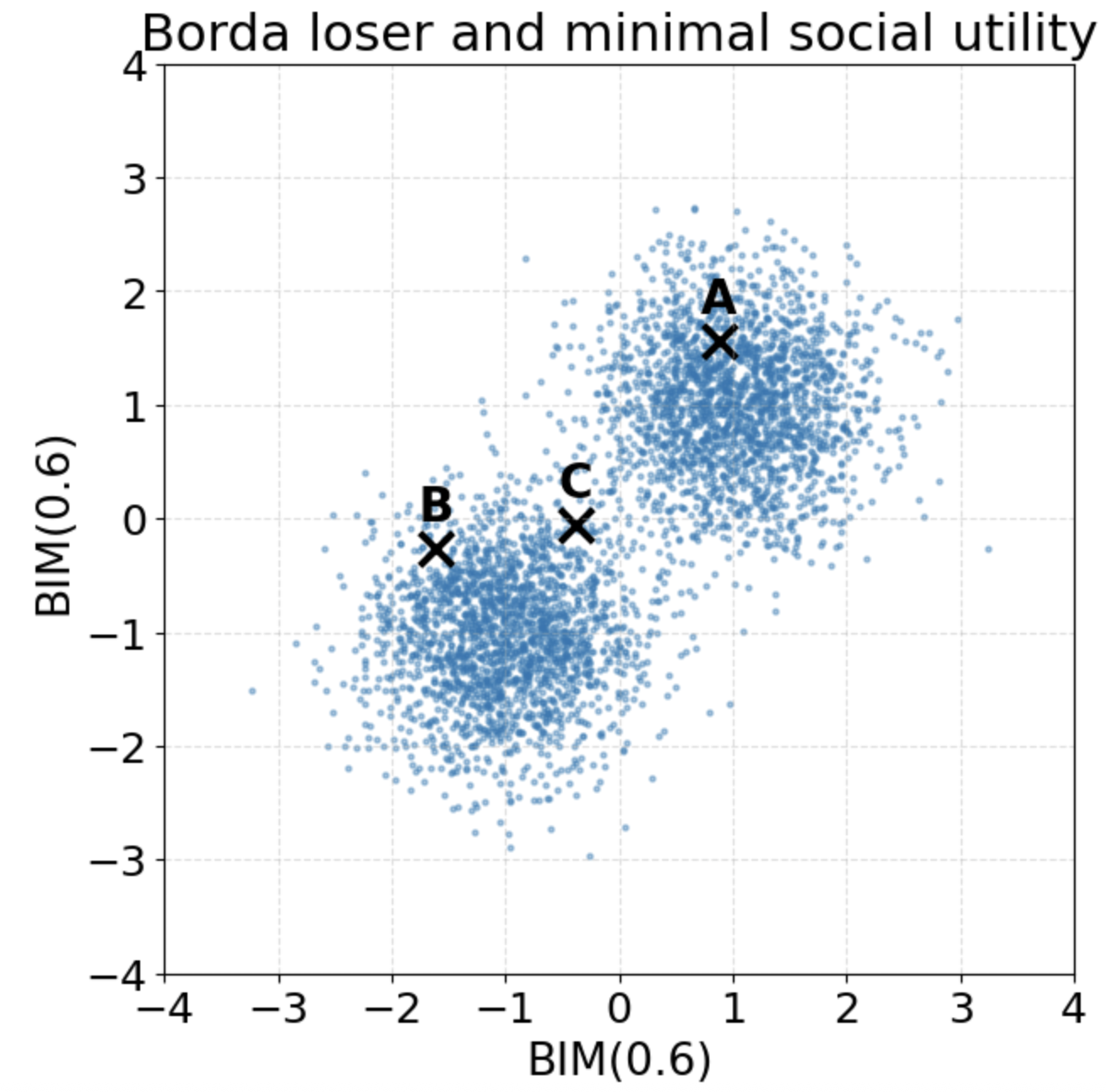}
  \end{subfigure}\hfill
  \begin{subfigure}[b]{0.48\textwidth}
    \centering
    \includegraphics[width=\linewidth, height=65mm]{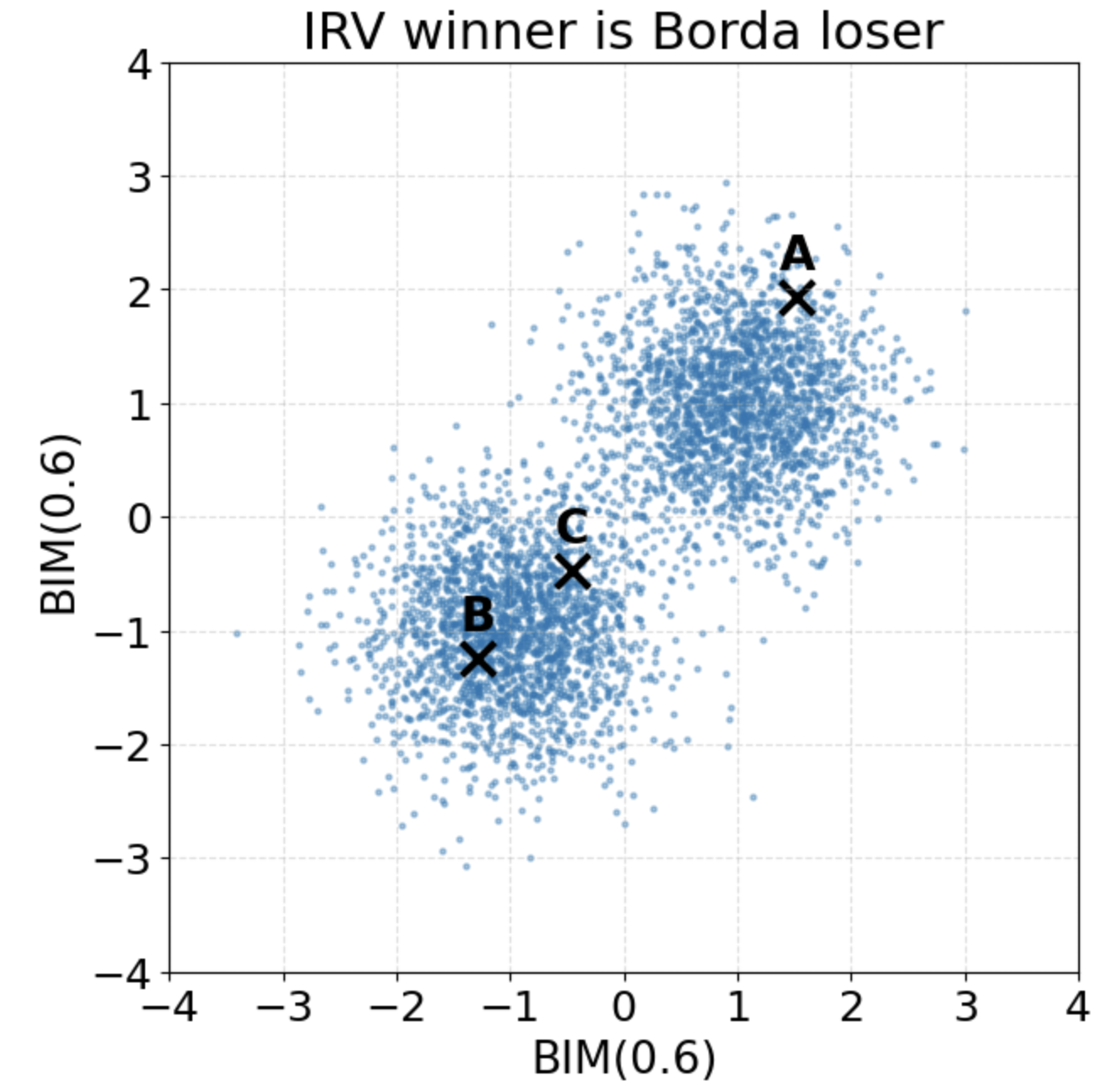}
  \end{subfigure}
  \caption{(Left) Two different elections generated under the model BIM(0.6) $\times$ BIM(0.6) where the IRV winner $B$ is the Borda loser. In the election on the Left, $B$ is also the candidate that  minimizes social utility.}\label{figure:2D_center_squeeze}
\end{figure}

\clearpage
\section*{Appendix 2: IC calculations }

\begin{proposition}  (i) Under the IC model, the limiting probability that the IRV winner coincides with the Borda loser is approximately 2.811\%. (ii) The limiting probability that the IRV winner coincides with the Bucklin loser is approximately 7.174\%.

\end{proposition}
\vskip 2mm

\noindent \textbf{Proof} 
We continue the proof of (i), following the numbered steps from the text.

\vskip 2mm 
\noindent \textbf{Step 3} Finding $P(A\succ B \succ C, B \succ_C A, BSc(A)>BSc(B), \mbox{ and }BSc(B)>BSc(C))$

This requires the  inequalities:
\begin{align*}
a_1+a_2 &> b_1+b_2 \quad  \mbox{ with normal vector } \mathbf v_1 = (1, 1, -1, -1, 0, 0), \\
b_1+b_2 &> c_1+c_2 \quad  \mbox{ with normal vector } \mathbf v_2 = (0,0, 1, 1, -1, -1),\\
b_1+b_2+c_2&>a_1+a_2+c_1  \quad  \mbox{ with normal vector } \mathbf v_3 = (-1,-1, 1, 1, -1, 1), \\
a_1+a_2+\frac{1}{2}(b_1+c_1)  & > b_1+b_2+\frac{1}{2}(a_1+c_2)\quad  \mbox{ with normal vector } \mathbf v_4 = (\frac{1}{2}, 1, -\frac{1}{2}, -1, \frac{1}{2}, -\frac{1}{2}),
\end{align*}
and the additional condition
$$ b_1+b_2+\frac{t}{2}(a_1+c_2)  > c_1+c_2+\frac{t}{2}(a_2+b_2)\quad   \mbox{ with normal vector }  \mathbf v_5 = (\frac{t}{2},  -\frac{t}{2}, 1, 1- \frac{t}{2}, -1,  \frac{t}{2}-1).$$
Again,  if $t=0$, this reduces to a previous condition; if $t=1$, we recover the condition that $BSc(B)>BSc(C)$.
Finally, we add the   orthogonal vector
$\mathbf v_6=(1, 1, 1, 1, 1, 1).$

As before, we obtain 
\begin{align*}
\alpha_{15}& = \cos^{-1}\left(-\frac{(1, 1, -1, -1, 0, 0) \cdot  (\frac{t}{2},  -\frac{t}{2}, 1, 1- \frac{t}{2}, -1,  \frac{t}{2}-1)  }{ 2\sqrt{4-2t+t^2}}\right)=\cos^{-1}\left(\frac{4-t}{4\sqrt{4-2t+t^2}}\right), \\
\alpha_{25} &= \cos^{-1}\left(-\frac{ (0,0, 1, 1, -1, -1) \cdot  (\frac{t}{2},  -\frac{t}{2}, 1, 1- \frac{t}{2}, -1,  \frac{t}{2}-1)  }{ 2\sqrt{4-2t+t^2}}\right)=\cos^{-1}\left(\frac{t-4}{2\sqrt{4-2t+t^2}}\right), \\
\alpha_{35}&= \cos^{-1}\left(-\frac{(-1,-1, 1, 1, -1, 1) \cdot  (\frac{t}{2},  -\frac{t}{2}, 1, 1- \frac{t}{2}, -1,  \frac{t}{2}-1)}{\sqrt{6}\sqrt{4-2t+t^2}}\right)=\cos^{-1}\left(-\frac{\sqrt{2}}{\sqrt{3}\sqrt{4-2t+t^2}}\right), \\
\alpha_{45}&= \cos^{-1}\left(-\frac{(\frac{1}{2}, 1, -\frac{1}{2}, -1, \frac{1}{2}, -\frac{1}{2}) \cdot  (\frac{t}{2},  -\frac{t}{2}, 1, 1- \frac{t}{2}, -1,  \frac{t}{2}-1)}{\sqrt{3}\sqrt{4-2t+t^2}}\right)=\cos^{-1}\left(\frac{\sqrt{3}}{2\sqrt{4-2t+t^2}}.\right)
\end{align*}
So
\begin{align*}
 d\alpha_{15}& = \frac{ \sqrt{3}t }{ (t^2-2t+4) \sqrt{5t^2-8t+16}  },\\
  d\alpha_{25}& =\frac{- \sqrt{3} }{ (t^2-2t+4) },\\
   d \alpha_{35}&= \frac{ \sqrt{2}(1-t)}{(t^2-2t+4)\sqrt{3t^2-6t+10}}, \mbox{  and}\\
    d \alpha_{45}&= \frac{ \sqrt{3}(t-1)}{(t^2-2t+4)\sqrt{4t^2-8t+13}}.
\end{align*}

Let $P_{ijkl}$ be the vertex lying at  the intersection of $S_i$, $S_j$, $S_k$ and $S_l$. A similar calculation to that done previously gives $P_{1235}=(-1, 1, 2, -2, 0, 0) $, $P_{1245}=(-1, 1, 1, -1, -1, 1)$, $P_{1345}=(t-2, 2, t-2, 2, -t, -t)$, and $P_{2345}=(1, 1, -2, 1, -2, 1).$ Each set of three of these points defines a triangle on the surface of the unit sphere in  $R^3$. Following the methodology of \cite{KM}, we calculate the volume  of these triangles as follows.  

Beginning with the triangle with vertices $P_{1235}, P_{1245}$ and $P_{1345}$, let $\beta_1= \angle(P_{1235}, P_{1245}), \beta_2= \angle(P_{1235}, P_{1345}),$ and $ \beta_3= \angle(P_{1245}, P_{1345})$. 
Then
\begin{align*}
\cos(\beta_1)&=\frac{(-1, 1, 2, -2, 0, 0)  \cdot (-1, 1, 1, -1, -1, 1)}{\sqrt{10}\cdot \sqrt{6}}=\frac{2}{2\sqrt{10}}=\frac{\sqrt{3}}{\sqrt{5}}, \\
\cos(\beta_2)&=\frac{(-1, 1, 2, -2, 0, 0) \cdot (t-2, 2, t-2, 2, -t, -t) }{\sqrt{10}\cdot 2\sqrt{t^2-2t+4}}=\frac{t-4}{2\sqrt{10}\sqrt{t^2-2t+4}}, \mbox{  and}\\
\cos(\beta_3)&=\frac{(-1, 1, 1, -1, -1, 1) \cdot (t-2, 2, t-2, 2, -t, -t)}{\sqrt{6}\cdot 2 \sqrt{t^2-2t+4}} =0.\end{align*}
If the angle at  $P_{ijkl}$ is denoted  $\gamma_{ijkl}$, this implies
\begin{align*}
\cos(\gamma_{1235}) &= \frac{\cos(\beta_4 )-\cos(\beta_1)\cos(\beta_2)}{\sin(\beta_1)\sin(\beta_2)}  = \frac{0 - (\frac{\sqrt{3}}{\sqrt{5}} )(\frac{t-4}{2\sqrt{10}\sqrt{t^2-2t+4}})}
{(\frac{\sqrt{2}}{\sqrt{5}})(\frac{\sqrt{3}\sqrt{13t^2-24t+48}}{2\sqrt{10}\sqrt{t^2-2t+4}})}=\frac{4-t}{\sqrt{2}\sqrt{13t^2-24t+48}},\\
\cos(\gamma_{1245}) &= \frac{\cos(\beta_2)-\cos(\beta_1)\cos(\beta_4)}{\sin(\beta_1)\sin(\beta_4)}  = \frac{\frac{t-4}{2\sqrt{10}\sqrt{t^2-2t+4}}-( \frac{\sqrt{3}}{\sqrt{5}})(0)}{(\frac{\sqrt{2}}{\sqrt{5}})(1)}=\frac{t-4}{4\sqrt{t^2-2t+4}}, \mbox{  and}  \\
\cos(\gamma_{1345}) &= \frac{\cos(\beta_1)-\cos(\beta_2)\cos(\beta_4)}{\sin(\beta_2)\sin(\beta_4)}  = \frac{\frac{\sqrt{3}}{\sqrt{5}} -(\frac{t-4}{2\sqrt{10}\sqrt{t^2-2t+4}} )(0)}{\frac{\sqrt{3}\sqrt{13t^2-24t+48}}{2\sqrt{10}\sqrt{t^2-2t+4}}(1)}=\frac{2\sqrt{2}\sqrt{t^2-2t+4}}{\sqrt{13t^2-24t+48}}.
\end{align*}
Thus, by  the Gauss-Bonnet Theorem, 
\begin{equation*}
Vol_{2} (S_1 \cap S_5) = \cos^{-1}( \frac{4-t}{\sqrt{2}\sqrt{13t^2-24t+48}} )+\cos^{-1}( \frac{t-4}{4\sqrt{t^2-2t+4}} )+\cos^{-1}( \frac{2\sqrt{2}\sqrt{t^2-2t+4}}{\sqrt{13t^2-24t+48}} )-\pi. \end{equation*}

In a similar way, we can find expressions for  the volume of the triangle  with vertices $P_{1235}$, $P_{1245}$ and $P_{2345},$ 
\begin{equation*}
Vol_{2} (S_2 \cap S_5) = \cos^{-1}( \frac{3}{\sqrt{14}}  ) +\cos^{-1}( -\frac{\sqrt{3}}{2}  )+\cos^{-1}( \frac{\sqrt{6}}{\sqrt{7}} )-\pi, \end{equation*}
 the volume of the triangle  with vertices  $P_{1235}$, $P_{1345}$ and $P_{2345}$, 
\begin{align*}
Vol_{2} (S_3 \cap S_5) = \cos^{-1}( \frac{t+6}{\sqrt{7}\sqrt{13t^2-24t+48} } )&+\cos^{-1}(\frac{ -4t^2+7t-12}{\sqrt{13t^2-24t+48}  \sqrt{4t^2-8t+13}}  )\\
& \hskip 1cm+\cos^{-1}(  \frac{t-1}{\sqrt{7}\sqrt{4t^2-8t+13}} )-\pi, \end{align*}
and the volume of the triangle  with vertices  $P_{1245}$, $P_{1345}$ and $P_{2345},$ 
\begin{equation*}
Vol_{2} (S_4 \cap S_5) = \cos^{-1}( \frac{\sqrt{3}}{2\sqrt{t^2-2t+4}}\  )+\cos^{-1}( 0  )+\cos^{-1}( 0 )-\pi. \end{equation*}

 By the Sch{\"a}fli formula (with $n=4$),  we get
\begin{align*}
Vol_4(S_{t=1}) & = Vol_{4}(S_{t=0}) +\frac{1}{3} \int_{0}^1  Vol_{2}(S_1 \cap S_5) d\alpha_{15}+Vol_{2}(S_2 \cap S_5) d\alpha_{25} \\
& \quad \quad \quad \quad \quad+Vol_{2}(S_3 \cap S_5) d\alpha_{35} +Vol_{2}(S_4 \cap S_5) d\alpha_{45} \\
& =  Vol_{4}(S_{t=0}) +\frac{1}{3}\left[I_1+I_2+I_3+I_4 \right]
 \end{align*}
 where, from Step 2, $Vol_{4}(S_{t=0}) = 0.01342486879$,  and
 
\begin{align*} 
I_1 &= \int^1_0 \left[ \cos^{-1}( \frac{4-t}{\sqrt{2}\sqrt{13t^2-24t+48}} )+\cos^{-1}( \frac{t-4}{4\sqrt{t^2-2t+4}} ) \right. \\
& \left. +\cos^{-1}( \frac{2\sqrt{2}\sqrt{t^2-2t+4}}{\sqrt{13t^2-24t+48}} )-\pi \right]  \frac{ \sqrt{3}t }{ (t^2-2t+4) \sqrt{5t^2-8t+16}  } dt \\
& 
\approx 0.05475823646,\\
 I_2 &=\int^1_0 \left[ \cos^{-1}( \frac{3}{\sqrt{14}}  )+\cos^{-1}( -\frac{\sqrt{3}}{2}  )+\cos^{-1}( \frac{\sqrt{6}}{\sqrt{7}} )-\pi  \right]  \frac{- \sqrt{3} }{ (t^2-2t+4) }  dt \\
& \approx -0.2641661714,   \\
 I_3 &=\int^1_0  \left[  \cos^{-1}( \frac{t+6}{\sqrt{7}\sqrt{13t^2-24t+48} } )+\cos^{-1}(\frac{ -4t^2+7t-12}{\sqrt{13t^2-24t+48}  \sqrt{4t^2-8t+13}}  )  \right. \\
& \left. +\cos^{-1}(  \frac{t-1}{\sqrt{7}\sqrt{4t^2-8t+13}} )-\pi   \right] \frac{ \sqrt{2}(1-t)}{(t^2-2t+4)\sqrt{3t^2-6t+10}} dt \\
 & 
 \approx 0.1238208804,\\
 I_4 & =\int^1_0 \left[  \cos^{-1}( \frac{\sqrt{3}}{2\sqrt{t^2-2t+4}}\  )+\cos^{-1}( 0  )+\cos^{-1}( 0 )-\pi  \right]  \frac{ \sqrt{3}(t-1)}{(t^2-2t+4)\sqrt{4t^2-8t+13}} dt \\
&  \approx -0.08221263827. 
 \end{align*}
 Thus, the probability is approximately
\begin{align*}
 0.01342486879& +\frac{1}{3((8/3)\pi^2)} \left[0.05475823646 -0.2641661714 \right.\\
& \hskip 1cm \left.+0.1238208804-0.08221263827 \right]=0.01129966085.
 \end{align*}

\vskip 2mm 
 
\noindent \textbf{Step 4} Finding $P(A\succ B\succ C, B \succ_C  A, BSc(A)>Bc(B), \mbox{ and }BSc(C)>BSc(B))$. 

Subtracting the probability in Step 3 from the probability in Step 2, we obtain
\begin{align*}
P(A\succ B\succ &C, B \succ_C  A, BSc(A)>Bc(B), \mbox{ and }BSc(C)>BSc(B)) \\
& = P(A\succ B\succ C, B \succ_C A, BSc(A)>BSc(B)) \\
& \hskip 1cm -P(A\succ B\succ C, B \succ_C A, BSc(A)>BSc(B), \mbox{ and } BSc(B)>BSc(C)))\\
& \approx 0.01342486879-0.01129966085=0.00212520794.
\end{align*}

\noindent \textbf{Step 5} Sketching the conclusion of the proof of (i). From Step 3, the probability that $B$ is the IRV winner and the Borda loser is approximately $0.00212520794$. The probability that $A$ is the IRV winner and the Borda loser is determined similarly and is approximately equal to  $0.00255956339$.  Thus, the probability that the IRV winner is the Borda loser is equal to
$$0.00212520794 + 0.00255956339 = 0.00468477133.$$
Since the order of the first-place votes $A \succ B \succ C$ was arbitrary, we obtain a final answer by multiplying by 6 to get
$$6 \times 0.00468477133 = 0.02810862798$$
or roughly 2.811\%.

\noindent \textbf{ Proof of (ii)} The proof of (ii) is entirely analogous to the proof of (i). Note that since no candidate achieves a majority, the Bucklin loser is the candidate with the smallest number of first and second-place votes. This is equivalent to the candidate with the most last-place votes. 

Assuming that the order of first-place votes is $A \succ B\succ C$, we use a similar argument to  find that the probability that $B$ is the IRV winner and the Bucklin loser is approximately 0.0024298360 and the probability that $A$ is the IRV winner and the Bucklin loser is approximately 0.0095329879. Thus, the probability that the IRV winner is the Bucklin loser is
$$6 \times (0.0024298360+ 0.0095329879) = 0.0717769434,$$
or approximately 7.178\%.

\end{document}